\documentclass[runningheads]{llncs}
\usepackage[utf8]{inputenc}
\usepackage{xcolor}
\usepackage{hyperref}
\usepackage{todonotes}
\usepackage{extarrows}
\usepackage{authblk}
\definecolor{bluegray}{RGB}{160,200,200}
\usepackage[notion,quotation]{knowledge}

\definecolor{Blue Sapphire}{HTML}{005f73} 
\definecolor{Gamboge}{HTML}{ee9b00}
\definecolor{Ruby Red}{HTML}{9b2226}

\IfKnowledgePaperModeTF{
}{
	\knowledgestyle{intro notion}{color={Ruby Red}, emphasize}
	\knowledgestyle{notion}{color={Blue Sapphire}}
	\hypersetup{
		colorlinks=true,
		breaklinks=true,
		linkcolor={Blue Sapphire}, 
		citecolor={Blue Sapphire}, 
		filecolor={Blue Sapphire}, 
		urlcolor={Blue Sapphire},
	}
}
\IfKnowledgeCompositionModeTF{
	\knowledgeconfigure{anchor point color={Ruby Red}, anchor point shape=corner}
	\knowledgestyle{intro unknown}{color={Gamboge}, emphasize}
	\knowledgestyle{intro unknown cont}{color={Gamboge}, emphasize}
	\knowledgestyle{kl unknown}{color={Gamboge}}
	\knowledgestyle{kl unknown cont}{color={Gamboge}}
}{
}

\knowledge{notion}
| witnesses
| witnessing
| witnessed

\knowledge{notion}
| Lock-sharing system
| lock-sharing system
| LSS

\knowledge{notion}
|controllable
|Controllable

\knowledge{notion}
|uncontrollable
|Uncontrollable

\knowledge{notion}
| global configuration

\knowledge{notion}
| control strategy
| strategy
| strategies
| Strategies
| Strategy

\knowledge{notion}
| local strategy
| local strategies

\knowledge{notion}
| uses two locks
| 2LSS

\knowledge{notion}
| locally live

\knowledge{notion}
| regular objective
| regular objectives
| Regular objectives

\knowledge{notion}
| process-fair
| fair@proc
| \textbf{process-fair}
| Fairness

\knowledge{notion}
| fair@act
| action-fair

\knowledge{notion}
| yields a $P$-deadlock
| deadlock@P

\knowledge{notion}
| yields a partial deadlock
| partial deadlock
| partial deadlocks

\knowledge{notion}
| circular deadlock
| circular deadlocks

\knowledge{notion}
| pattern@inf
| infinitary@pat
| infinitary patterns
| infinitary pattern
| Infinitary patterns

\knowledge{notion}
| weak pattern@inf
| weak@inf 

\knowledge{notion}
| sound
| Soundness

\knowledge{notion}
| exclusive
| Exclusive

\knowledge{notion}
| neutral

\knowledge{notion}
| stair decomposition

\knowledge{notion}
| stair pattern

\knowledge{notion}
| strong pattern@inf
| strong patterns@inf
| strong@inf

\knowledge{notion}
| switching

\knowledge{notion}
| global deadlock
| global deadlocks

\knowledge{notion}
| G_{\mathit{Inf}}

\knowledge{notion}
| \text{\sc{Blocks}}

\knowledge{notion}
| pattern
| Patterns
| patterns
| finite@pat
| finitary patterns
| Finitary patterns
| finitary@pat
| finitary pattern

\knowledge{notion}
| trace
| Trace

\knowledge{notion}
| strong pattern
| strong patterns
| strong@pat

\knowledge{notion}
| weak pattern
| weak patterns
| weak@pat

\knowledge{notion}
| risky

\knowledge{notion}
| nested

\knowledge{notion}
| admits a pattern

\knowledge{notion}
| behavior@two

\knowledge{notion}
| regular verification problem
| Regular verification problem

\knowledge{notion}
| partial deadlock objective

\knowledge{notion}
| generalised deadlock objective

\knowledge{notion}
| partial deadlock problem

\knowledge{notion}
| process deadlock problem
| Process deadlock problem
| process deadlocks

\knowledge{notion}
| generalised deadlock control problem

\knowledge{notion}
| deterministic Emerson-Lei automaton
| deterministic Emerson-Lei automata
| Emerson-Lei automaton
| ELA
| DELA
\newif\ifappendix
\newrobustcmd\introinrestatable[1]{%
	\ifappendix%
	\kl{#1}%
	\else%
	\intro{#1}%
	\fi%
}

\usepackage[utf8]{inputenc}
\usepackage[T1]{fontenc}
\usepackage{mathtools}
\usepackage{amssymb}
\usepackage{thm-restate}
\usepackage{xspace}
\usepackage{tikz}
\usepackage{array}
\usepackage{changepage}
\usepackage{hyperref}
\usepackage{enumitem}
\usepackage{xcolor}
\usepackage{algorithm}
\usepackage{algpseudocode}
\definecolor{light-gray}{gray}{0.75}
\definecolor{darkgreen}{RGB}{0,200,0}
\definecolor{lightpurple}{RGB}{220,0,220}
\usepackage{algpseudocode}
\usepackage[capitalise]{cleveref}
\usepackage{cite}

\usetikzlibrary{arrows,calc,automata,shapes,positioning}
\tikzset{AUT style/.style={>=angle 60,initial text= ,every edge/.append style={thick},every state/.style={thick,minimum size=15,inner sep=0.5}}}

\renewcommand{\d}{\delta}

\newcommand{\g}{\gamma}
\renewcommand{\l}{\lambda}

\renewcommand{\S}{\Sigma}
\renewcommand{\epsilon}{\varepsilon}
\renewcommand{\phi}{\varphi}

\newcommand{\pow}[1]{2^{#1}}
\newcommand{\nats}{\mathbb{N}}
\newcommand{\size}[1]{|#1|}
\newcommand{\set}[1]{\{ #1 \}}

\newcommand{\tss}{\Ss}

\newcommand{\run}{w}
\newcommand{\get}[1]{\mathtt{get}_{#1}}
\newcommand{\rel}[1]{\mathtt{rel}_{#1}}
\newcommand{\rec}{\text{\sc{rec}}}
\newcommand{\send}{\text{\sc{send}}}
\newcommand{\acqs}{\text{\sc{acqs}}}
\newcommand{\rels}{\text{\sc{rels}}}
\newcommand{\start}{\text{\sc{start}}}
\newcommand{\NOP}{\text{\sc{nop}}}

\newcommand{\edge}[1]{\xrightarrow{#1}}

\newcommand{\greenedge}[1]{\xhookrightarrow{#1}}
\newcommand{\rededge}[1]{\xmapsto{#1}}

\newcommand{\PP}{\mathbb{P}}
\newcommand{\Pat}{\mathbf{pat}}
\knowledgenewrobustcmd{\BTPP}{BT_{\PP}}
\knowledgenewrobustcmd{\strongedge}[1]{\cmdkl{\greenedge{#1}}}
\knowledgenewrobustcmd{\weakedge}[1]{\cmdkl{\rededge{#1}}}

\knowledgenewrobustcmd{\spat}{\cmdkl{\Longrightarrow}}
\knowledgenewrobustcmd{\wpat}{\cmdkl{\dashrightarrow}}
\knowledgenewrobustcmd{\spatinf}{\cmdkl{\Longrightarrow}}
\knowledgenewrobustcmd{\wpatinf}{\cmdkl{\dashrightarrow}}
\knowledgenewrobustcmd{\Gp}{\cmdkl{G_{\PP}}}
\knowledgenewrobustcmd{\Gu}{\cmdkl{G_{u}}}
\newcommand{\Ginf}{\kl{G_{\mathit{Inf}}}}

\knowledgenewrobustcmd{\BTu}{\cmdkl{BT_{u}}}
\knowledgenewrobustcmd{\FTu}{\cmdkl{FT_{u}}}
\knowledgenewrobustcmd{\PZ}{\cmdkl{Proc_Z}}
\knowledgenewrobustcmd{\pad}[1]{#1^{\cmdkl{\dummy}}}
\knowledgenewrobustcmd{\systemchoices}{\cmdkl{\mathbb{SC}}}
\knowledgenewrobustcmd{\Gblockchain}[2]{\cmdkl{\mathcal{G}}_{#1,#2}}

\makeatletter
\newcommand{\dashover}[2][\mathop]{#1{\mathpalette\df@over{{\dashfill}{#2}}}}
\newcommand{\fillover}[2][\mathop]{#1{\mathpalette\df@over{{\solidfill}{#2}}}}
\newcommand{\df@over}[2]{\df@@over#1#2}
\newcommand\df@@over[3]{%
	\vbox{
		\offinterlineskip
		\ialign{##\cr
			#2{#1}\cr
			\noalign{\kern1pt}
			$\m@th#1#3$\cr
		}
	}%
}
\newcommand{\dashfill}[1]{%
	\kern-.5pt
	\xleaders\hbox{\kern.5pt\vrule height.4pt width \dash@width{#1}\kern.5pt}\hfill
	\kern-.5pt
}
\newcommand{\dash@width}[1]{%
	\ifx#1\displaystyle
	2pt
	\else
	\ifx#1\textstyle
	1.5pt
	\else
	\ifx#1\scriptstyle
	1.25pt
	\else
	\ifx#1\scriptscriptstyle
	1pt
	\fi
	\fi
	\fi
	\fi
}
\newcommand{\solidfill}[1]{\leaders\hrule\hfill}
\makeatother

\knowledgenewrobustcmd{\Owns}{\cmdkl{\text{\sc{Owns}}}}
\newcommand{\Blocks}{\kl{\text{\sc{Blocks}}}}
\knowledgenewrobustcmd{\Inf}{\cmdkl{\text{\sc{Inf}}}}
\newcommand{\OwnsN}{\text{\sc{Owns}}^N}
\newcommand{\InfN}{\text{\sc{Inf}}^N}
\newcommand{\varinf}{\text{inf}}

\newcommand{\ownfunc}{\emph{owns}}

\newcommand{\aut}{\mathcal{A}}
\newcommand{\objaut}[1]{\mathcal{B}_{#1}}
\knowledgenewrobustcmd{\lss}{\emph{"LSS"}}

\newcommand{\lang}[1]{\mathcal{L}(#1)}
\knowledgenewrobustcmd{\trace}[1]{\cmdkl{tr}(#1)}

\newcommand{\init}{\mathit{init}}

\newcommand\restrict[2]{{
		\left.\kern-\nulldelimiterspace
		#1
		\vphantom{\big|}
		\right|_{#2}
}}

\newcommand{\es}{\emptyset}

\newcommand{\act}[1]{\xlongrightarrow{#1}}

\newcommand{\nop}{\mathit{nop}}
\newcommand{\dummy}{\square}

\renewcommand{\d}{\delta}

\newcommand{\Proc}{\mathit{Proc}}

\newcommand{\Ss}{\mathcal{S}}

\newcommand{\tuple}[1]{\langle #1 \rangle}

\newcommand{\PSPACE}{\text{\sc Pspace}}
\newcommand{\PTIME}{\text{\sc Ptime}}
\newcommand{\NP}{\text{\sc NP}}

\title{Model-checking lock-sharing systems against regular constraints}
\titlerunning{Model-checking LSS}
\author{Corto Mascle}
\institute{LaBRI, Université de Bordeaux\\
	\email{corto.mascle@labri.fr}\\
	\url{https://corto-mascle.github.io/}}

\begin{document}

\appendixfalse

\maketitle

\begin{abstract}
	We study the verification of distributed systems where processes are finite
	automata with access to a shared pool of locks. We consider objectives that
	are boolean combinations of local regular constraints.
	We show that the problem, PSPACE-complete in general, falls in NP with the right assumptions on the system. We use restrictions on the number of locks a process can access and the order in which locks can be released. 
	We provide tight complexity bounds, as well as a subcase of interest that can be solved in PTIME.
	
\end{abstract}

\section{Introduction}

Concurrent programs often prove more challenging to verify than sequential ones, as the state space explodes easily, unless processes follow very closely what the others are doing or have completely decorrelated executions. Verification of such programs can be traced back to the work of Taylor~\cite{Taylor1983}, and has been the subject of a variety of approaches, which reflect the numerous possible modelisations of distributed systems.
Looking for an error trace is typically \PSPACE-hard when processes are
finite-state systems, i.e., the cost of exploring an exponential number of
configurations.
The reason is that most models of concurrent programs, be it with
rendez-vous, message passing, or shared variables, can encode the problem of
deciding whether a set of deterministic finite automata have a common accepted
word. This is the case for instance for the classical model of Zielonka automata~\cite{Zielonka87}.

We study lock-sharing systems (LSS for short), a simple model for concurrent programs using mutexes. Processes have access to a pool of locks. Each process is represented by an automaton whose transitions acquire and release locks.
Locks restrict the behaviours of the system, as a process cannot take a lock already held by another process.
Similar systems were considered by Gupta, Kahlon and Ivanci{\'{c}} in \cite{KahIvaGup05}, with only two processes, each being a pushdown system.
They proved that the verification of regular constraints relating local runs was undecidable, and provided a fine-grained analysis of the decidable cases in that paper and later ones \cite{Kahlon09, KahGup06lics}. They also showed that detecting deadlocks is decidable under some restrictions.
This exact approach contrasts with other ones, such as in
\cite{BouajjaniET2003} or \cite{QadeerR2005}, which tackle more general systems but use
approximations of the set of possible runs. Chapter 18 of~\cite{ClarkeHVB2018} gives an overview of those works.

We consider the verification problem for the model studied in~\cite{GimMMW22}. That paper focused of synthesizing local strategies to avoid global deadlocks. 
Here we consider a much larger family of properties: boolean combinations of
local regular properties.
Unlike~\cite{GimMMW22} we do not discuss the synthesis problem, but the model-checking problem.

In this work we present an analysis of restrictions on lock-sharing systems that suffice in order to obtain more tractable complexities than \PSPACE. We mainly focus on two restrictions, 2LSS and nested LSS. The first one demands that each process only accesses two different locks, the second one that each process takes and releases locks as if they were stored in a stack: they can only release the lock taken latest.
Several works already showed the interest of the nested restriction to obtain
tractable verification problems, see for instance
\cite{Brotherson21,KahIvaGup05}. The contribution of~\cite{Brotherson21} consists in an \NP~algorithm (and an implementation) for detecting deadlocks (more specifically, configurations where some subset of processes is blocked as they all try to acquire locks held by other processes of that subset) in concurrent programs. They use a syntax for programs that can be translated to what we call sound nested exclusive LSS in this paper.  
As for the systems with two locks per process, they can already exhibit a variety of behaviours. Dijkstra's famous dining philosophers problem matches this constraint.
These restrictions  have a common point: local runs can be summarised in short descriptions, called \emph{patterns}. Patterns contain enough information to determine whether local runs can be interleaved to form a global run. Some form of patterns for finite runs of nested systems, called acquisition history, was already considered in~\cite{KahIvaGup05}, but was only focused on systems with two processes and with no considerations of complexity. In \cite{GimMMW22}  patterns are defined on finite runs and used to compute local strategies to prevent deadlocks in LSS. 
We show that we can extend the techniques to handle much larger classes of specifications, in the framework of verification.

In order to do this, we extend the notion of patterns to infinite runs and provide necessary and sufficient conditions on patterns to represent runs that can be interleaved into a (fair) global run.
This allows us to verify the system against local specifications by first
guessing for each process a pattern, checking compatibility of these patterns
and then checking individually in each process the existence of a bad run with
the corresponding pattern.
Thus we avoid exploring the product of all processes. 

This approach yields \NP~algorithms for the verification of (boolean combinations of) local specifications for 2LSS and nested LSS.
With an additional constraint, called exclusiveness, we even obtain a \PTIME~algorithm for some specific objective called process deadlock, which requires one given process to be forever unable to advance after some point in the run.

We provide matching lower bounds for these results. 
In general our problem is~\PSPACE-complete, even with a bounded number of locks per process. 
It is \NP-complete in the nested case, even with exclusiveness, for some weak objectives (the hardness proof solves a question left open by the authors in~\cite{Brotherson21}), and a bounded number of locks per process.
As for 2LSS, the problem is \NP-complete as well.
Furthermore, those lower bounds make little use of the specification, proving that the complexities are in some sense inherent to the systems. 

\paragraph{Overview}
In Section~\ref{sec:def} we recall some definitions and give some intuition about the global framework. Then in Section~\ref{sec:patterns} we generalise the notion of patterns that was used in \cite{GimMMW22} (Definition~14), after which we present the results that we are able to obtain through this technique:
In Section~\ref{sec:part-deadlock} we discuss a particular specification, for which the problem can be solved in \PTIME~for exclusive systems, and provide an \NP~lower bound when we do not assume exclusiveness. 
In Section~\ref{sec:reg-objectives} we prove the \PSPACE-completeness of the general problem and contrast it with its \NP-completeness in the 2LSS case.
Finally, in Section~\ref{sec:nest}, we prove that the verification of nested systems is \NP-complete, with a very robust lower bound, that survives exclusiveness, weak objectives, and even a bounded number of locks.

\section{Definitions}
\label{sec:def}

First we recall the definition of a lock-sharing system

\begin{definition}[Lock-sharing system]
	Let $\Proc$ be a finite set of processes.
	
	A ""lock-sharing system"" ("LSS" for short) $\Ss=((\aut_p)_{p\in\Proc},T, op)$ is given by
	a family of transition systems, a set $T$ of locks, and a function $op$ described below.
	
	Each transition system $\aut_p$ is given as a tuple $(S_p, \S_p, \delta_p, init_p)$ with $S_p$ a finite set of states, $init_p$ the initial state, $\S_p$ a finite alphabet and $\delta_p : S_p \times \S_p \to S_p$ a \textbf{partial} function.
	We require that the $\S_p$ are pairwise disjoint, and define $\S = \bigcup_{p \in \Proc} \S_p$.
	
	Consider a set of operations $Op(T) = \set{\get{t}, \rel{t}, nop \mid t \in T}$. The function $op : \Sigma \to Op(T)$ associates with each letter of $\S$
	an operation on locks.
	For all $p \in \Proc$ we define $T_p = \set{t \in T \mid \exists a \in \S_p, op(a) = \get{t}}$ the set of locks $p$ may acquire. 
	
	\AP A ""2LSS"" is an LSS where every $T_p$ has two elements.
\end{definition}

\begin{remark}
	In \cite{GimMMW22}, the transition functions $\delta_p$ output a pair $(s, op)$ with a state and an operation. Here we will assume without loss of generality that the operation of a transition is 
	determined by its action; we can use $\Sigma \times Op(T)$ as our alphabet 
	instead of just $\Sigma$ and thus explicitly describe the sequence of operations
	in the actions. 
\end{remark}

We fix an "LSS" $\Ss=((\aut_p)_{p\in\Proc},T, op)$ for the rest of this section.

A \emph{local configuration} of process $p$ is a state from $S_p$ together
with the locks $p$ currently owns: $(s,B)\in S_p\times 2^{T_p}$. 
The initial configuration of $p$ is $(\init_p,\es)$, namely the initial state with no locks.
A transition between configurations $(s, B) \xrightarrow{a} (s',
B')$ exists when $\d_p(s,a)=s'$ and one of the following holds:
\begin{itemize}
	\item $op(a) = \nop$ and  $B = B'$;
	
	\item $op(a)=\get{t}$,  $t \notin B$ and $B' = B\cup \set{t}$;
	
	\item $op(a)=\rel{t}$, $t \in B$, and $B'= B \setminus
	\set{t}$.
\end{itemize}
A \emph{local run} $a_1 a_2 \cdots$ of $\aut_p$ is defined as a finite or infinite sequence over
$\S_p$ such that there exists a sequence of local configurations
$(\init_p,\es)=(s_0,B_0)\act{a_1}_p (s_1,B_1)\act{a_2}_p\cdots$ (we will specify explicitly when we talk about local runs that do not start in the initial configuration).

\AP We say that a finite local run $\run_p = a_1\cdots a_n$ is ""neutral"" if for all $1 \leq i \leq n$ such that $op(a_i) = \get{t}$ for some $t \in T$, there exists $j > i$ such that $op(a_j) = \rel{t}$. Equivalently, the configuration obtained after executing $\run_p$ is in $S_p \times \set{\es}$.

\AP A \intro{global configuration} is a tuple of local configurations
$C=(s_p, B_p)_{p\in\Proc}$ provided the sets $B_p$  are pairwise
disjoint:
$B_{p}\cap B_{q}=\es$ for $p\not=q$. 
This is because \textbf{a lock can be taken by at most one process at a time}. 
The initial configuration is the tuple of initial configurations of
all processes.

Runs of such systems are \emph{asynchronous}, with transitions between two
consecutive configurations done by a single process:
$C\act{(p,a)}C'$ if $(s_p,B_p)\act{a}_p(s'_p,B'_p)$ and
$(s_q,B_q)=(s'_q,B'_q)$ for every $q\not=p$.
A global run is a sequence of transitions between global
configurations.
Since our systems are deterministic we usually identify a global run with the sequence
of transition labels. 
A global run $w$ \emph{determines a local run} of each process: $w|_p$ is the subsequence
of $p$'s actions in $w$. We also say that $w|_p$ is the projection of $w$ on $p$.

In what follows we will assume that each process keeps track in its state of
the set of locks it owns.
Note that this assumption does not compromise the complexity results provided there
is a bound on the number of locks a process can access: the number of states is
then multiplied by a constant factor.

\begin{definition}
	\label{def:soundness}
	A process of an "LSS" is ""sound"" if its transition system $\aut_p$ keeps track of the set of locks it has in its states. 
	Formally, let $\aut_p = (S_p, \delta_p, init_p)$, $p$ is sound if there exists a function $\ownfunc_p : S_p \to \pow{T_p}$ such that:
	\begin{itemize}
		\item for all local runs $\run = a_1 a_2 \cdots a_n$ ending in a state $s$,
		we have $(init_p, \es) \xrightarrow{a_1} \cdots \xrightarrow{a_n} (s, \ownfunc_p(s))$.
		
		\item for all states $s \in S_p$, there is no outgoing transition of $s$ that acquires a lock in $\ownfunc_p(s)$ or releases a lock that is not in $\ownfunc_p(s)$.
	\end{itemize}
	 
	An "LSS" is sound if all its processes are.
\end{definition}

Note that this property can be easily checked on a given "LSS": it suffices to set $\ownfunc(init_p)$ to $\es$, apply a DFS to compute candidates for $\ownfunc(s)$ for all states, and then check consistency of $\ownfunc$ with respect to each transition.

%
%
%


We want to be able to define deadlocks in terms of languages of runs. To this end, we have to restrict our attention to process-fair runs, in which every process is either blocked after some point 
or executes an action infinitely many times.
This is often called strong fairness in the literature. This way if a process stops doing anything after some point in a run, it means it is blocked.

\begin{definition}
	 A run $\run$ is called ""process-fair"" if for all $p \in \Proc$, either $\run$ contains infinitely many actions of $\S_p$, or there is a point after which no action of $p$ can ever be executed at any moment in the run.
	 
	 \AP We say that a process-fair run yields a \intro{global deadlock} if it is finite,
	 i.e., at some point there are no actions that can be executed in any of the
	 processes, and the system cannot advance any more. Note that a process-fair run is finite if and only if it yields a global deadlock.
	 
	 \AP We say that a process-fair run yields a \intro{partial deadlock} if its projection on one of the $\S_p$ is finite, i.e., after some point one of the processes is never able to execute any action.
	
\end{definition}

In all that follows we will have to work with finite and infinite words 
simultaneously as "LSS" executions may be finite or infinite. 
We will use a dummy letter $\dummy$, and finite runs will be padded with an
infinite suffix $\dummy^\omega$ so that we can express objectives as languages
of infinite words.

\AP From now on we will write $\intro*\pad{u}$ for the padded version of a word $u$, i.e.,

\[\pad{u} = 
\begin{cases}
	u &\text{ if } u \text{ is infinite}\\
	u\dummy^\omega &\text{ if } u \text{ is finite.}
\end{cases}
\]



%
%
%

We will now define the set of properties we want to verify.
This class of objectives is inspired by Emerson-Lei automata, introduced in~\cite{EmersonL1987}, which we will use for several proofs of upper bounds. 
Note that we will use non-deterministic Emerson-Lei automata, while our objectives are expressed using deterministic automata.

\begin{definition}
	An ""Emerson-Lei automaton"" ("ELA" for short) is a tuple $\aut = (S, \Sigma, \Delta, init, \phi)$ 
	with $S$ a finite set of states, 
	$\Sigma$ a finite alphabet, 
	$\Delta : S\times \Sigma \times S$ a transition function, 
	$init \in S$ the initial state and 
	$\phi$ a boolean formula over variables $ \set{\varinf_s \mid s \in S}$.
	
	Such an automaton recognises a language $\lang{\aut} \subseteq \Sigma^\omega$.
	An infinite word $w$ is accepted if there is a run of $w$ in $\aut$ such that  $\phi$ is satisfied by the valuation
	evaluating $\varinf_s$ to $\top$ if and only if $s$ appears infinitely often
	in the run.
\end{definition}

Our objectives are defined in a similar fashion, but with one automaton per process and a single formula expressing a condition on which states (among the ones of all automata) are seen infinitely often.

\begin{definition}
	\label{def:reg-obj}
	A \intro{regular objective} is a pair $((\objaut{p})_{p \in \Proc},
	\phi)$ such that each $\objaut{p}$ is a deterministic automaton with
	a set of states $S_{\objaut{p}}$ over the alphabet $\Sigma_p \cup \set{\dummy}$, and
	$\phi$ is a boolean formula over the set of variables $\set{\varinf_{p,s} \mid p
	\in \Proc, s \in S_{\objaut{p}}}$.
	
	\AP Let $\run$ be a "\textbf{process-fair}" run, and for each $p$ let $\run_p$ be its projection on $\S_p$.
	We say that $\run$ satisfies a regular objective $((\objaut{p})_{p \in \Proc},
	\phi)$ if $\phi$ is satisfied by the valuation evaluating $\varinf_{p,s}$ to
	$\top$ if and only if the unique run of $\lang{\objaut{p}}$ on $\pad{\run_p}$ 
	goes through $s$ infinitely many times.
\end{definition}

We argue that these specifications are quite expressive and at the
same time allow us to stay in reasonably low complexity classes. 

\paragraph*{Regular objectives are expressive.} They can  express
properties such as reachability (with local or global
configurations) or safety,
as well as properties related to deadlocks, such as "partial deadlock" or "global
deadlock": As we focus on "process-fair" runs, a local projection of a run is finite if and only if the corresponding process is blocked at some point and has no available action for the rest of the run.
Hence, we can express for instance a "global deadlock" with an objective requiring the local run of every process $p$ to be finite.

Moreover, the flexibility of boolean formulas allows us to relate configurations between processes: say each process has to decide between $0$ and $1$, then we can express agreement by demanding that they all select $0$ or all $1$.

Regular objectives are furthermore closed under boolean combinations. 
They can be complemented by simply taking the negation of the formula $\phi$, and intersected in polynomial time by taking the product automaton for each process and adapting the formula.

%
%

\paragraph*{Complexity blows up quickly with more expressive objectives} 

"Regular objectives" only restrict the shape of local runs without any requirement on their interleaving. Restrictions on interleavings would lead to \PSPACE-hardness very quickly. 
As we will see in Section~\ref{sec:reg-objectives}, as soon as we can have a system where processes are required to synchronize in some way, we also obtain \PSPACE-hardness.

Objectives that are sensitive to interleavings of local runs can be used to test the emptiness of the intersection of
languages of $n$ DFAs, even without any locks. 
We can take $\Proc = \set{p_1, \ldots, p_n}$ and $\S_p = \set{a_p, b_p, c_p}$
for all $p$ and ask for a global run in $(a_{p_1}\cdots a_{p_n} + b_{p_1}\cdots
b_{p_n})^*(c_{p_1}\cdots c_{p_n})^\omega$ in the LSS constructed from those DFAs.
%

%
%
%
%

%

In this work we study the problem of finding a run satisfying some given "regular
objective".

\begin{definition}
	We define the ""regular verification problem"" as:
	
	Input: a "sound" "LSS" $\tss$ and a "regular objective" $RO = ((\objaut{p})_{p \in \Proc}, \phi)$
	
	Output: Is there a "process-fair" run of $\tss$ satisfying $RO$?
\end{definition}

Note that we define the problem existentially: we are looking for a bad run, hence the given objective should express the set of runs that we want to avoid. 
We use this formulation as it simplifies a bit our proofs, and as "regular
objectives" are easy to complement.

We also define the problem in the particular case of "process deadlocks": 
Here, we ask whether there is a run in which some given process $p$ is eventually blocked forever. 
We define it as our standard example of a ``simple'' objective. We will show
that we can decide it in \PTIME~in a particular case, and we will use it for
complexity lower bounds, thus showing  that those complexities are already
inherent to the systems.

\begin{definition}
	We define the ""process deadlock problem"" as:
	
	Input: a "sound" "LSS" $\tss$ and a process $p$.
	
	Output: Is there a  "process-fair" run of $\tss$ whose projection on $p$ is finite?
%
\end{definition}

As our last definition in this part, we introduce exclusive "LSS", in which a process that can acquire a lock cannot do any other operation from
the same state.

\begin{definition}[\intro{Exclusive}]
	A process is exclusive if its transition system $\aut_p$ is such that for all states $s$, if $s$ has an outgoing transition acquiring some lock $t$, then all other outgoing transitions acquire that same lock $t$. 
	An "LSS" is exclusive if all its processes are.
\end{definition}

%
%
%
%
%
%

\section{Patterns for 2LSS}
\label{sec:patterns}

In this section we define patterns for "2LSS". These are summaries of bounded size of 
the operations executed during a run, which contain enough information to 
tell if local runs can be combined into a global one. Let us first define a couple of useful functions over local runs.

%
%
%
%
%
%

\begin{definition}
	Given a finite local run $\run_p = a_0 \cdots a_n$ of a process
    $p$,   
    we define $\intro*\Owns(\run_p)$ as the set of locks $p$ holds after executing $op(a_0) \cdots op(a_n)$.
    
    We extend the function $\Owns$ to infinite runs by setting $\Owns(a_1 a_2 \cdots)$ as the set of locks kept indefinitely by $p$ after some point. Formally, we define 
     $\Owns(a_1 a_2 \cdots) = \bigcup_{i \in \nats} \bigcap_{j > i} \Owns(a_0\cdots a_j)$.
    
    The \intro{trace} of an infinite run $\run_p = a_1 a_2 \cdots$, denoted by $\intro*\trace{\run_p}$, is the infinite word $A_0 A_1
        \cdots \in (2^{T})^\omega$ with $A_i = \Owns(a_1 \cdots a_i)$ the set of locks held by $p$ after executing the first $i$ actions of $\run_p$.
        
        We also define $\intro*\Inf(\run)$ as the set of sets of locks that $p$ owns infinitely often when executing $\run_p$:
        \[\Inf_p(a_1 a_2 \cdots) = \set{A \subseteq T_p \mid A=\Owns(a_1 \cdots a_i) \text{ for infinitely many } i}\]
\end{definition}

We start with patterns of \emph{finite runs} as in~\cite{GimMMW22}.
We redefine them here with a formalism adapted to our purpose.

\begin{definition}[\intro{Finitary patterns}]
	\label{def:finpat}
	\label{def:patterns2locks}
	Finitary patterns are defined for finite local runs of a
    "2LSS".
    Let $p$ be a process, $T_p = \set{t_1, t_2}$ its locks.
    Let $\run = a_1 a_2 \cdots a_n$ be a finite local run of $p$.	
	The pattern of $\run$ is defined as the set $\Owns(\run)$ along with an information on its strength:
	
	\begin{itemize}
		\item If $\Owns(\run) = \set{t_1}$ (resp.~$\set{t_2}$) and the last
		operation on locks in $\run$ is $\rel{t_2}$ (resp.~$\rel{t_1}$) then we say that $\run$ has a \intro{strong pattern}, denoted as $\fillover{\Owns(\run)}$
		
		\item Otherwise we say that $\run$ has a \intro{weak pattern}, denoted $\dashover{\Owns(\run)}$  
	\end{itemize}
\end{definition}

In \cite{GimMMW22} the "global deadlock" problem was studied, so only patterns of finite runs were of interest. 
We define patterns of infinite runs as we have to
account for the runs of processes that do not get blocked.

\begin{definition}[""Infinitary patterns""]
	\label{def:infpat}
	Let $\run$ be an infinite local run of a process $p$ accessing
    locks $T_p = \set{t_1, t_2}$. 
    Let $\trace{\run} = A_0 A_1 \cdots \in
    (2^{T_p})^\omega$.
    The pattern of $\run$ is given by $\Inf(\run)$ along with an information on its strength:    
	\begin{itemize}
		\item \AP We say that $w$ has a \intro(inf){strong pattern} $\fillover{\Inf(\run)}$ when $\trace{\run} \in
		(2^{\set{t_1,t_2}})^*\set{t_1,t_2}\set{t_1}^\omega$ (the process has both locks at some point, releases one of them and does not do any other operation on locks afterwards).
		Observe that in this case $\Inf(\run) = \set{\set{t_1}}$.
		
		\item Otherwise, $\run$ has the \intro(inf){weak pattern} $\dashover{\Inf(\run)}$. 
	\end{itemize}
	
	\AP We say that $\run$ is ""switching"" if $\emptyset \notin \Inf(\run)$ and $\Owns(\run) = \emptyset$.
	This means that eventually, 
	$p$ never releases both locks, but releases each one infinitely often.
	In particular, $T_p \in \Inf(\run)$.
\end{definition}

\begin{figure}[ht]
	\begin{tikzpicture}[node distance=1.8cm,auto,>= triangle
	45,scale=.6]
	\tikzstyle{initial}= [initial by arrow,initial text=,initial
	distance=.7cm, initial where= left]
	\tikzstyle{accepting}= [accepting by arrow,accepting text=,accepting
	distance=.7cm,accepting where =right]
	
	\node[state, minimum size=15pt,initial] (s1) at (0,0) {1};
	\node[state, minimum size=15pt] (s2) at (3,0) {2};
	\node[state, minimum size=15pt] (s3) at (6,0) {3};
	\node[state, minimum size=15pt] (s4) at (10,-1) {4};
	\node[state, minimum size=15pt] (s5) at (14,0) {5};
	\node[state, minimum size=15pt] (s6) at (10,1) {6};
	
	\path[->] 	
	(s1) edge node[above] {$\get{t_1}$} (s2)
	(s2) edge node[above] {$\get{t_2}$} (s3)
	;
	\path[->, bend right=20] 
	(s3) edge node[below left] {$\rel{t_1}$} (s4)
	(s4) edge node[below right] {$\get{t_1}$} (s5)
	(s5) edge node[above right] {$\rel{t_2}$} (s6)
	(s6) edge node[above left] {$\get{t_2}$} (s3)
	;
\end{tikzpicture}
	\caption{A process with a single infinite run whose pattern is $\dashover{\set{\set{t_1}, \set{t_2}, \set{t_1,t_2}}}$ ("switching").
		It also has finite runs of patterns $\dashover{\es}$, $\dashover{\set{t_1}}$, $\dashover{\set{t_1, t_2}}$, $\fillover{\set{t_2}}$ and $\fillover{\set{t_1}}$}
	\label{fig:example-switching}
\end{figure}
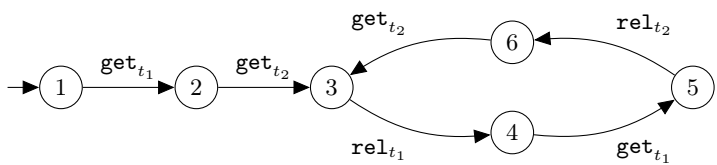

\begin{example}
	Consider the process $p$ displayed in Figure~\ref{fig:example-switching}. 
	It has a single infinite run, which eventually cycles between states 4 (in which it has only $t_2$), 6 (in which it has only $t_1$), and 3 and 5 (in which it has both), hence it has as "infinitary pattern" $\dashover{\set{\set{t_1}, \set{t_2}, \set{t_1, t_2}}}$, i.e., it is "switching".
	
	This system is "sound", i.e., for all finite runs $\run$, $\Owns(\run)$ is determined by its end state. Furthermore the "pattern" is "strong@@pat" if $\Owns(\run)$ is a singleton and the last operation in $\run$ is a $\rel{}$, which is also determined by the end state in this system.
	We can infer that all runs ending in state 1 have pattern $\dashover{\es}$, in state 2 $\dashover{\set{t_1}}$, in state 3 and 5 $\dashover{\set{t_1, t_2}}$, in state 4 $\fillover{\set{t_2}}$, and in state 6 $\fillover{\set{t_1}}$. 
\end{example}

Note that for each of the patterns defined above, the set of runs 
matching that pattern is a regular language. Although this fact is clear, we formalise it in the following lemma. This allows us to give explicitly (very small) automata recognising those languages, and we think that the proof of this lemma may help the reader understand the relation between "finitary@@pat" and "infinitary patterns".

\begin{lemma}
	\label{lem:DELA-pat}
	Let $p \in \Proc$ be a process. 
	
	For each pattern $\Pat$ described in Definitions~\ref{def:finpat} and~\ref{def:infpat} we can define a (deterministic) "ELA" $\aut_{\Pat}^p$ (with 12 states) over the alphabet $\Sigma_p \cup \set{\dummy}$ recognizing the language consisting of $\pad{\run}$ with $\run$ a local run of $p$ whose pattern is $\Pat$.
\end{lemma}

\begin{proof}
	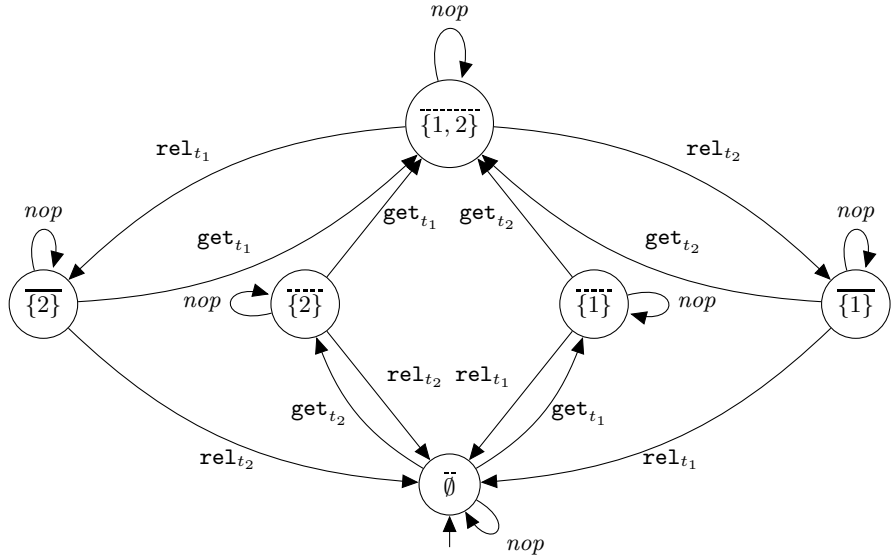
\begin{figure}
		\begin{tikzpicture}[node distance=1.8cm,auto,>= triangle
	45,scale=.6]
	\tikzstyle{initial}= [initial by arrow,initial text=,initial
	distance=.7cm, initial where= below]
	\tikzstyle{accepting}= [accepting by arrow,accepting text=,accepting
	distance=.7cm,accepting where =right]
	
	\node[state,initial] (es) at (0,0) {$\dashover{\es}$};
	\node[state] (1w) at (3.2,4) {$\dashover{\set{1}}$};
	\node[state] (1s) at (9,4) {$\fillover{\set{1}}$};
	\node[state] (2w) at (-3.2,4) {$\dashover{\set{2}}$};
	\node[state] (2s) at (-9,4) {$\fillover{\set{2}}$};
	\node[state] (12) at (0,8) {$\dashover{\set{1,2}}$};

	\path[->, bend left=20] 	
	(es) edge node[left] {$\get{t_2}$} (2w)
	(1s) edge node[below] {$\rel{t_1}$} (es)
	(1s) edge node[above right] {$\get{t_2}$} (12)
	(12) edge node {$\rel{t_2}$} (1s)
	;
	\path[->, bend right=20] 
	(es) edge node[right] {$\get{t_1}$} (1w)
	(2s) edge node[below] {$\rel{t_2}$} (es)
	(2s) edge node {$\get{t_1}$} (12)
	(12) edge node[above left] {$\rel{t_1}$} (2s);
	\path[->]
	(2w) edge node[right] {$\get{t_1}$} (12)
	(2w) edge node {$\rel{t_2}$} (es)
	(1w) edge node[above left] {$\rel{t_1}$} (es)
	(1w) edge node[left] {$\get{t_2}$} (12)
	(es) edge [out=330, in=300, looseness=8] node {$\nop$} (es)
	(1w) edge [loop right] node {$\nop$} ()
	(1s) edge [loop above] node {$\nop$} ()
	(2w) edge [loop left] node {$\nop$} ()
	(2s) edge [loop above] node {$\nop$} ()
	(12) edge [loop above] node {$\nop$} ()
	
	;
\end{tikzpicture}
		\caption{The automaton structure for pattern recognition. Every state $s$ has a transition $\dummy$ to its copy $\pad{s}$, with a $\dummy$ self-loop, which is not displayed.}
		\label{fig:pat-recognition}
	\end{figure}
	
	For all patterns we use the same states and transitions, which keep track of the "finitary patterns". 
	They are described in Figure~\ref{fig:pat-recognition} with $T_p = \set{t_1, t_2}$. 
	We labelled edges with operations instead of actions as the transitions of an action $a$ depend only on $op(a)$ here.
	For each state $s$ we have a transition labelled by $\dummy$ leading to a copy $\pad{s}$ of that state with only a self-loop labelled by $\dummy$.
	The desired pattern is then expressed as an Emerson-Lei condition to obtain an "ELA". 
	
	For a "finitary pattern" $\Pat$ the formula $\varinf_{\pad{\Pat}}$ suffices, to indicate that the automaton read a run of pattern $\Pat$ and then only $\dummy$.
	
	For an "infinitary pattern" such that $\Inf(\run_p) = \set{\set{t}}$ for some $t$, we have to distinguish "strong@@inf" and "weak@@inf".
	If the pattern is strong we use the formula $\varinf_{\fillover{\set{t}}} \land \bigwedge_{s \neq \fillover{\set{t}}} \neg \varinf_s$ saying that we stay in state $\fillover{\set{t}}$ indefinitely, otherwise we use $\varinf_{\dashover{\set{t}}}\bigwedge_{s \neq \dashover{\set{t}}} \neg \varinf_s$ saying that we stay in $\dashover{\set{t}}$ indefinitely.
	
	Otherwise we only have to check the set of sets of locks owned infinitely often, hence we use the formula $\bigwedge_{J \in \Inf(\run_p)} \phi_J \land \bigwedge_{J \notin \Inf(\run_p)} \neg \phi_J$,
	where $\phi_J$ is $\varinf_{\dashover{J}} \lor \varinf_{\fillover{J}}$ if $J$ is a singleton, and $\varinf_{\dashover{J}}$ otherwise.
\end{proof}

We now present the key proposition on patterns for "2LSS". 
It states when a set of local runs can be combined into a global run.
Note that the criterion depends only on the patterns of the local runs and the last
states they reach.
This will be the crucial ingredient in the proof that the "regular verification problem" is in \NP~for "2LSS".

\begin{restatable}{proposition}{CHARACPATTERNS}
	\label{prop:charac-schedulable-pat}
	Consider a family of local runs $(\run_p)_{p \in \Proc}$ (each of them can be finite or infinite).
	
	\AP We write $\introinrestatable{G_{\mathit{Inf}}}$ for the undirected graph whose vertices are locks and with a $p$-labelled edge between $t_1$ and $t_2$ whenever $T_p= \set{t_1, t_2}$ and $\run_p$ is "switching".
	
	For all finite $w_p$ let $s_p$ be its end state. We define the set of locks that can be acquired from $s_p$: $\introinrestatable{\text{\sc{Blocks}}}_p = \set{t \mid \exists a, op(a) = \get{t} \text{ and }
		\delta_p(s_p,a) \text{ is defined}}$.
	
	The local runs $(\run_p)_{p \in \Proc}$ can  be scheduled into a "process-fair" global run if and only if the following conditions are all satisfied.

	\begin{enumerate}
		\item \label{C1} If $\run_p$ is finite then all outgoing transitions from its
		end state $s_p$ acquire a lock.  
		
		\item \label{C3} All sets $\Owns(\run_p)$ are disjoint.
		
		\item \label{C4} All $\Blocks_p$ are included in  $ \bigcup_{p \in \Proc} \Owns(\run_p)$.
		
		\item \label{C5} The intersection $ \Owns(\run_p) \cap \bigcup_{J \in \Inf(\run_q)} J$ is
		empty for all pairs of processes $p \neq q$ such that $\run_q$ is infinite.
%
%
	\item \label{C6} There is a total order $\leq$ on locks such that for all $p$
	whose run $\run_p$ has a strong pattern $\fillover{\set{t_1}}$ (for finite runs) or
		$\fillover{\set{\set{t_1}}}$ (for infinite runs) we have $t_1 \leq t_2$; where $t_2$ is
		the other lock used by  $p$.

		\item \label{C7} There is no process $p$ such that (1) $\set{t, t'} \in \Inf(\run_p)$
		and (2) there is a path in $\Ginf$ between $t$ and $t'$ not using a $p$-labelled edge.
		In particular, $\Ginf$ is acyclic.
	\end{enumerate}
\end{restatable}

\begin{proof}
	{\Large \textbf{$\Rightarrow$:}} We start with the left-to-right implication. Let $(\run_p)_{p \in \Proc}$ be a family of local runs, 
	suppose they can be scheduled into a "process-fair" global run $\run$.
	
	For all finite local runs $\run_p$, as $\run$ is "process-fair", 
	after some point $p$ cannot ever execute any action.
	
	As a consequence, $s_p$ (the state reached after executing $\run_p$) cannot have any outgoing
	transition executing $\rel{}$ or $\nop$, as those can always be executed.  
	The locks of $\Blocks_p$ are never free after some point, as otherwise $p$
	would be enabled infinitely often on the run, so the run would not be
	"process-fair". 
	This shows condition~\ref{C1}.
	
	%
	All finite runs $\run_p$ stop while holding the locks of $\Owns(\run_p)$.
	All infinite $\run_p$ eventually acquire and never release the locks 
	of their $\Owns(\run_p)$.
	Hence the $\Owns(\run_p)$ sets need to be pairwise disjoint, proving condition~\ref{C3}.
	
	Furthermore, if a lock is not in some $\Owns(\run_p)$ then 
	it is free infinitely often, 
	and thus cannot be in $\Blocks_p$ for any $p$, as $\run$ is "process-fair".
	This proves condition~\ref{C4}.
	
	All locks of $\bigcup_{J \in \Inf(\run_q)} J$ are held by $q$ infinitely often, hence they cannot be in any $\Owns(\run_p)$ with $p \neq q$, which shows 
	condition~\ref{C5}.
	
	If a run $\run_p$ of a process $p$ using locks $t_1,t_2$ has a "pattern"
	$\fillover{\set{t_1}}$ or $\fillover{\set{\set{t_1}}}$ 
	then the last operation on $t_1$ (when $p$ acquires it for the last time) is followed by at least 
	one operation on $t_2$ in the run $\run$. 
	We satisfy condition~\ref{C6} by setting $\leq$ as an order on locks 
	such that $t \leq t'$ whenever $t$ is only acquired finitely many times and there is an operation on $t'$ 
	after the last operation on $t$ in $\run$ .
	
	We demonstrate condition~\ref{C7} by contradiction.
	Say there exist such locks and process, i.e., there exist $t = t_1, \ldots, t_n = t'$ and $p_1, \ldots, p_{n-1}$ without $p$ such that 
	for all $1 \leq i < n$, $p_i$ accesses $t_i$ and $t_{i+1}$ 
	and $\run_{p_i}$ is "switching". 
	Then all $p_i$ are always holding a lock after some point.
	
	As $\set{t_n, t_1} \in \Inf(\run_p)$, this means that $p$ holds $t_n$ and $t_1$ 
	simultaneously infinitely often. 
	Whenever that happens, processes $p_1,\dots,p_{n-1}$ have to share the remaining 
	($n-2$) locks, hence one of them holds no lock, 
	contradicting the fact that $\emptyset \notin \Inf(\run_{p_i})$ for all $i$.
	
	{\Large \textbf{$\Leftarrow$:}} For the other direction, 
	suppose $(\run_p)_{p \in \Proc}$ satisfies all the conditions of the list.
	We construct a "process-fair" global run 
	whose local projections are the $\run_p$.
	
	To do so, we will construct a sequence of finite runs $v_0, v_1, \ldots$ such that $v_0 v_1 \cdots$ is such a global run.
	
	We will ensure that the following property is satisfied for all $i \in \nats$:
	
	\begin{align}\label{prop-vi}\nonumber
		& \text{For all processes } p, \text{ after executing } v_0\cdots v_i,\\
		& \text{ if } \Owns(\run_p) \in \Inf(\run_p) \text{ then }  \Owns((v_0\cdots v_i)|_p) = \Owns(\run_p)\\
		&\nonumber \text{ otherwise } w_p \text{ is "switching" and } p \text{ holds one lock.}
	\end{align}
	
	We will also make sure that all $p$ with an infinite $\run_p$ execute an action in infinitely many $v_i$.
	
	The first run $v_0$ has to be constructed separately as we require it to satisfy some extra conditions.
	We construct $v_0$ such that for all $p$: 
	\begin{itemize}
		\item If $\run_p$ is finite then $v_0|_p = \run_p$.
		
		\item If $\run_p$ is infinite then $\run_p = v_0|_p u_p$ 
		with $u_p$ such that for every prefix $u'_p$ of $u_p$, $\Owns(v_0|_p u'_p) \in \Inf(\run_p)$.
		Furthermore if $\emptyset \in \Inf(\run_p)$ then $\Owns(v_0|_p) = \emptyset$.
		
		In other words, we execute a prefix of each infinite run 
		such that what follows matches its asymptotic behaviour.
	\end{itemize}
	
	\paragraph{Construction of $v_0$} 	
	\begin{itemize}
		\item First, for all infinite $\run_p$ such that $\emptyset \in \Inf(\run_p)$, 
		there exist arbitrarily large finite prefixes of $\run_p$ 
		ending with $p$ having no lock. 
		Hence we can select one of those prefixes $v_0|_p$, 
		large enough for $p$ to never hold a set of locks 
		not in $\Inf(\run_p)$ later in the run.
		We execute all such $v_0|_p$ at the start. All locks are free afterwards.
		
		\item We then execute for all other $p$ with "weak patterns", 
		their maximal prefix ending with $p$ having no lock.
		All locks are still free.

		\item Then we execute all $\run_p$ with "strong patterns", 
		in increasing order according to $\leq$ (see condition~\ref{C6}) 
		on the locks $t_p$ such that $\Owns(\run_p)=\set{t_p}$. 
		We execute in full the finite ones, while for the infinite ones 
		we execute a prefix $v_0|_p$ such that in the end $p$ owns only $t_p$ and 
		never acquires the other lock afterwards (recall that $\Inf(\run_p) = \set{\set{t_p}}$ in that case).
		Say we executed some of those local runs, let $p$ be a process such that $\run_p$ has a "strong 
		pattern" accessing locks $t_1 \leq t_2$, say we want to execute $v_0|_p$. 
		By condition~\ref{C3}, all $\Owns(\run_p)$ are disjoint, hence there is no other 
		process $q$ with $\Owns(\run_q) = \set{t_1}$.
		The only locks that are not free at that point are the $t$ such that 
		$t < t_1$ and $\Owns(\run_q) = \set{t}$ for some $q$ with a "strong pattern".
		Therefore, both $t_1$ and $t_2$ are free, and $v_0|_p$ can be executed.
		In the end the aforementioned locks $t_p$ are taken and all others are free.

		\item  	Then we consider the finite $\run_p$ with non-empty $\Owns(\run_p)$ and "weak patterns". 
		For those, we can execute the rest of the run 
		(we already executed the maximal prefix leading to them holding no lock), 
		as all they do is take the locks in $\Owns(\run_p)$, 
		which are free by conditions~\ref{C3} and \ref{C5}.
		
		\item 	For the infinite $\run_p$ with non-empty $\Owns(\run_p)$ and "weak patterns", there are two possibilities: 
		
		\begin{itemize}
			\item The first is that $p$ eventually keeps the same set of locks forever 
			and never executes any more operations on locks. 
			Then its "trace" $\trace{\run_p}$ is of the form 
			either $(\pow{T_p})^* \emptyset \set{t_1}^\omega$ 
			or $(\pow{T_p})^*\set{t_1,t_2}^\omega$. 
			In that case clearly we can just execute the run until we reach a point 
			after which $p$ only ever owns $\Owns(\run_p)$ forever.
			We can do this as all locks taken so far are either in $\Owns(\run_q)$ 
			for some $q$ with finite $\run_q$ 
			or are in an element of some $\Inf(\run_q)$ for some $q$.
			Thus all locks from those $\Owns(\run_p)$ are free 
			by conditions~\ref{C3} and \ref{C5}.
			
			\item The other possibility is that 
			$\trace{\run_p} \in (\pow{T_p})^*(\set{t_1}^* \set{t_1, t_2})^\omega$ 
			with $\set{t_1, t_2} = T_p$ and $\Owns(\run_p) = \set{t_1}$. 
			This happens if $p$ ultimately holds one lock forever 
			and acquires and releases the other one infinitely many times.
			At that point all locks that are taken are in some $\Owns(\run_p)$, thus by condition \ref{C5} both locks of $p$ are free. 
			Hence we can execute enough steps of $\run_p$ to reach a point at which 
			$p$ holds only $t_1$ and will only hold sets of locks of $\Inf(\run_p)$ afterwards.
		\end{itemize}

		\item Finally we consider the infinite "switching" runs $\run_p$. 
		All those processes must have $T_p \in \Inf(\run_p)$, hence by condition \ref{C5} 
		all their locks are free.
		By condition \ref{C7}, $\Ginf$ is acyclic.
		We can therefore pick one of those processes $p$ and a lock $t$ such that 
		no other such process accesses $t$. 
		We execute $\run_p$ until 
		$p$ only owns $t$ will only own sets of locks of $\Inf(\run_p)$ afterwards.
		All locks of the other such $p$ are still free, hence we can iterate that step until we executed a prefix of each of those $p$.
	\end{itemize}

	We have constructed a finite run $v_0$ whose projection $v_0|_p$ on $\S_p$ is such that if $\run_p$ is finite then $\run_p = v_0|_p$ 
	and  if $\run_{p}$ is infinite then $v_0|_p$ is a prefix of $\run_p$ such that all local configurations seen later in the run are in $\Inf(\run_p)$. Moreover $v_0$ satisfies property~\ref{prop-vi}.
	
	We now construct the remaining parts of the run.
	If all $\run_p$ are finite then $v_0$ proves the lemma (we can set all other $v_i$ as $\epsilon$). Otherwise we must describe the rest of the "process-fair" global run whose projections are the $\run_p$.	
	We start with a small construction that will help us define the $v_i$.	
	
	Suppose we constructed $v_0, \ldots, v_i$ so that property~\ref{prop-vi} is satisfied for all $j \leq i$. 
	Now suppose some lock $t_0$ is not in any $\Owns(\run_p)$ 
	and is not free after executing $v_0 \cdots v_i$.
	Then there exists a "switching" run $\run_{p_1}$ with $t_0 \in T_{p_1}$.
	
	Let $t_1$ be the other lock of $p_1$, 
	say it is not free. 
	By property~\ref{prop-vi}, $p_1$ holds only one lock and thus does not hold $t_1$. 
	By condition~\ref{C5} $t_1$ cannot be in some $\Owns(\run_p)$, thus, again by property~\ref{prop-vi}, as $t_1$ is not free,
	there exists a "switching" $\run_{p_2}$ such that $t_1 \in T_{p_2}$.
	Let $t_2$ be the other lock of $p_2$.
	
	We construct this way a sequence of processes $p_1, p_2, \ldots$ 
	and of locks $t_0, t_1, \ldots$ such that 
	$T_{p_j} = \set{t_{j-1}, t_{j}}$ and $\run_{p_j}$ is "switching" for all $j$. 
	This sequence cannot be infinite as each $p_j$ labels an edge in $\Ginf$, 
	which is finite and acyclic.
	
	Hence there exists $k$ such that $t_k$ is free. 
	We can therefore execute $\run_{p_k}$ until $p_k$ holds $t_k$
	and not $t_{k-1}$, then execute $\run_{p_{k-1}}$ 
	until $p_{k-1}$ holds $t_{k-1}$ and not $t_{k-2}$, and so on
	until $t_1$ is free.
	
	Hence if a lock $t$ is not in any $\Owns(\run_p)$ but is not free 
	after executing $v_0 \cdots v_i$ then we can prolong the prefix run 
	so that $t$ is free and some lock from the same connected component 
	in $\Ginf$ is not. 
	For all such $t$ and $i$ we name this prolongation of the run $\pi_{t,i}$.
	
	Say we already constructed $v_0, \ldots, v_i$, and that property~\ref{prop-vi} is satisfied for all $j \leq i$. We construct $v_{i+1}$. 
	Let $p$ be either a process that never executed an action, 
	or if there are no such processes, 
	the process whose last action in $v_0\cdots v_i$ is the earliest.
	
	We prolong the current run so as to execute some actions of $\run_p$.
	If the next action of $\run_p$ applies an operation $\nop$ we can execute it right away.
	The next action cannot execute a $\rel{}$ operation:
	After executing $v_0$ all processes with infinite $\run_p$ only own sets of locks of $\Inf(\run_p)$.
	By property~\ref{prop-vi}, if $\run_p$ is "switching" then after executing $v_0\cdots v_i$ the process $p$ holds one lock and will not release it as it would be left with no lock and $\es \notin \Inf(\run_p)$. If $p$ is not "switching" then after executing $v_0\cdots v_i$ it holds $\Owns(\run_p)$ and cannot release any lock as it keeps those forever. 
	
	Hence we are left with the case where the next action of $p$ acquires a lock $t$.
	If $t$ is not free we apply $\pi_{t,i}$ to free it 
	(and block another lock of the same connected component of $\Ginf$). 
	Note that after executing $\pi_{t,i}$ all processes with "switching" runs still hold one lock, and the others have not moved.
	
	\begin{itemize}
		\item 	If $\run_p$ is "switching" then $p$ was already holding a lock $t'$, and it can then take $t$ and then run $\run_p$ until it holds only 
		one lock again, thus respecting property~\ref{prop-vi}.
		
		\item Otherwise $p$ was holding $\Owns(\run_p)$ (by Property~\ref{prop-vi}) 
		and we have to let him take $t$ and then continue
		until $p$ holds exactly $\Owns(\run_p)$ again.
		
		\begin{itemize}
			\item If we can do it right away we do so.
			
			\item Otherwise it means that $p$ needs its other lock $t'$ to reach that 
			next step, and that this lock is taken. 
			More precisely, it means that $\Owns(\run_p) = \emptyset$ and 
			$\emptyset, \set{t,t'} \in \Inf(\run_p)$.
			
			As $\set{t,t'} \in \Inf(\run_p)$, by condition~\ref{C7},
			$t$ and $t'$ are not in the same connected component of $\Ginf$. 
			Hence we can execute $\pi_{t',i}$, without locking $t$ back, as $\pi_{t,i}$ and $\pi_{t',i}$ use disjoint sets of locks and processes.
			
			This ensures that both $t$ and $t'$ are free, 
			which allows $p$ to take $t$ and proceed to the next point 
			at which it holds $\emptyset$.
		\end{itemize}
		In both cases we end up in a configuration where $p$ owns $\Owns_{p}$, all processes with "switching" runs hold exactly one lock, and the other processes did not move, thus respecting property~\ref{prop-vi}. 
	\end{itemize}

	We have constructed $v_{i+1}$, ensuring that property~\ref{prop-vi} is satisfied for $i+1$. 
	Furthermore $v_{i+1}|_p$ is non-empty for $p$ a process with infinite $\run_p$ which either never executed anything before or executed its last action the earliest. This ensures that all $p$ with infinite $\run_p$ execute infinitely many actions in $v_0 v_1 \cdots$.
	Hence we obtain a global run $v = v_0 v_1 \cdots$ such that for all $p$ we have $v|_p = \run_p$.
	
	Furthermore we ensured that $v$ is "process-fair" as all $p$ 
	with finite runs are blocked: 
	all such $\run_p$ lead to a state from which only locks of $\Blocks_p$ can be taken, by condition~\ref{C1}, 
	and by condition~\ref{C4} all $\Blocks_p$ are included in 
	$\bigcup_{p \in \Proc} \Owns(\run_p)$, 
	the set of locks that are never free from some point on.
	
	As a result, there exists a "process-fair" run 
	whose local projections are the $\run_p$, 
	proving the right-to-left implication.
\end{proof}

\begin{example}
	Consider two processes $p$ and $q$ with the same transition system, displayed in Figure~\ref{fig:example-switching}. 
	We can prove that all "process-fair" runs of those two will end in a "global
	deadlock" using "patterns".
	
	Say there is a run whose projection on one of them (say, $p$) is infinite, then that projection $\run_p$ has pattern $\dashover{\set{\set{t_1}, \set{t_2}, \set{t_1,t_2}}}$, meaning it will take and release both locks infinitely often without releasing both at the same time after some point.
	
	Then $q$ does not have a compatible run: It cannot have the same "infinitary pattern" by condition \ref{C6} of Proposition~\ref{prop:charac-schedulable-pat}. Furthermore, by condition~\ref{C4} it cannot have any finitary pattern besides $\dashover{\es}$.
	However, its only run with that pattern is the empty one, which ends in the initial state, from which there is a transition executing $\get{t_1}$, meaning that by condition~\ref{C3} we should have $t_1 \in \Owns(\run_p)$, which is not the case.
	Thus there cannot be such a run.
\end{example}

%

\section{Process deadlocks}
\label{sec:part-deadlock}

While the complexity lower bounds presented in this work are robust to many restrictions, we can still find some interesting properties that can be verified on some systems  in polynomial time.
In~\cite{GimMMW22} (Lemma~22 and Proposition~24) it was proven that verifying if a ``locally live'' strategy on a "2LSS" allows a run leading to a global deadlock (in which all processes are blocked) can be done in polynomial time. An immediate consequence of this is that verifying if a "sound" "2LSS" in which all states have at least one outgoing transition has a run yielding a global deadlock can be done in~\PTIME.

From the results in~\cite{GimMMW22} we can also extract the \NP-completeness of finding a global deadlock in a "2LSS" when we allow states with no outgoing transitions.

\subsection{A \PTIME~algorithm for exclusive 2LSS}

Here we are interested in a different problem, the "process deadlock problem".
We provide a polynomial-time algorithm based on a key lemma that lists the different ways a process can be blocked.

Let $\tss$ be a "sound" "exclusive" "2LSS" and $p$ a process of $\tss$.	

\begin{lemma}\label{lem:fair-exclusive}
	Let $(s_p, B_p)_{p \in \Proc}$ be a "global configuration" and for each process $p$ let $u'_p$ be a local run starting in $(s_p, B_p)$ and such that $u'_p$ is either infinite or leads to a state with no outgoing transitions.
	
	There exists a "process-fair" global run $\run$ from $(s_p, B_p)_{p \in \Proc}$ such that for all $p$ its projection $\run_p$ on $\S_p$ is a prefix of $u'_p$.
\end{lemma}

\begin{proof}
	We construct $\run$ by iterating the following step:
	For each $p$ we set $u'_p = v_p w_p$ with $v_p$ the prefix of $u'_p$ executed so far.
	We select uniformly at random a process $p \in \Proc$. If it can execute the first action of $w_p$ then we let it do so, otherwise we do nothing.
	
	We iterate this procedure indefinitely.
	This produces a (possibly finite) global run of the system such that its local projections are prefixes of the $u'_p$. We prove that it is "process-fair".
	
	Let $p \in \Proc$, assume that $p$ has an available action at infinitely many steps.
	As our "LSS" is "exclusive", whenever $p$ has an available action and is in some state $s$,
	either all outgoing transitions are executing an operation $\nop$ or $\rel{}$ (and thus can all be executed as the system is "sound"), or they all acquire the same lock $t$ (as the system is "exclusive").
	Hence if one outgoing transition can be executed , they all can and thus in particular the next action of $u'_p$ is available.
	As a result, $p$ can execute the next action of $u_p$ at infinitely many steps, and thus will progress infinitely many times in $u_p$.
	
	In conclusion, with this procedure we either reach a "global deadlock", or we always have an available action, implying that at least one process will be able to progress infinitely many times and that the resulting run $u$ is infinite. 
	In the latter case, all processes that can execute an action at infinitely many steps of the run will do so, proving that the run is "process-fair".
\end{proof}

\begin{definition}
	Define the graph $G$ whose vertices are locks and with an edge $t_1 \edge{p} t_2$ if and only if the process $p$ has a local run $\run_p$ ending in a state where all outgoing transitions acquire $t_2$ and such that $\Owns(\run_p) = \set{t_1}$. We say that $\run_p$ ""witnesses"" the edge $t_1 \edge{p} t_2$.
\end{definition}

\begin{lemma}
	\label{lem:always-one-weak}
	For all $p \in \Proc$, if there is a $p$-labelled edge $t_1 \edge{p} t_2$ in $G$  then either $t_1 \edge{p} t_2$ is "witnessed" by a run with a weak pattern or its reverse $t_2 \edge{p} t_1$ is in $G$ and is "witnessed" by a run with a weak pattern.
\end{lemma}

\begin{proof}
	As $p$ has an edge $t_1 \edge{p} t_2$ in $G$, there is a local run which acquires both locks of $p$ at the same time. Let $\run_p$ be such a run of minimal length. The last operation in $\run_p$ must be a $\get{}$, by minimality, hence $\run_p$ is of the form $\run'_p a$ with $op(a) = \get{t}$ for some $t \in \set{t_1, t_2}$.
	Furthermore, suppose the last operation in $\run'_p$ besides $\nop$ is a $\rel{}$, then there is a previous configuration in $\run_p$ in which $p$ holds both of its locks, contradicting the minimality of $\run_p$. Hence $\run'_p$ has a weak pattern, and it leads to a state where $p$ may acquire $t$, thus has to acquire $t$ as the system is "exclusive". Furthermore $p$ is then holding its other lock, therefore $\run'_p$ "witnesses" an edge in $G$.
\end{proof}

\begin{lemma}
	\label{lem:cond-cycle}
	If $p$ has a reachable transition acquiring some lock $t$ and there is a path from $t$ to a cycle in $G$ then there is a "process-fair" global run with a finite projection on $p$. 
\end{lemma}

\begin{proof}
	Let $t = t_0 \edge{p_1} t_1 \edge{p_2} \cdots \edge{p_k} t_k$ be such a path in $G$ and let $t_k = t'_1 \edge{p'_1} \cdots t'_n \edge{p'_n} t'_{n+1} = t'_1 = t_k$ be such  a cycle.
	
	For all $1 \leq i \leq k$ we choose a run $\run_i$ "witnessing" $t_{i-1} \edge{p_i} t_i$. Similarly for all $1 \leq j \leq n$ we choose a run $\run'_j$ "witnessing" $t'_{j} \edge{p'_i} t'_{j+1}$, and we choose it so that it has a weak pattern whenever possible.
	
	\paragraph{Case 1: }If there exists $j$ such that $\run'_j$ has a weak pattern, then we proceed as follows: Let $\run'_j = u_j v_j$ so that $u_j$ is the maximal "neutral" prefix of $\run_j$. We execute $u_j$, leaving all locks free.
	Let $m$ be the maximal index such that $t_m \in \set{t'_1, \cdots, t'_n}$. We execute all $\run_i$ in increasing order for $1 \leq i \leq m$.
	
	Then we execute $\run'_{j+1} \cdots \run'_n \run'_1 \cdots \run'_{j-1}$ and then $v_j$. Then we end up in a configuration where all $p_i$ with $i\leq m$ are holding $t_{i-1}$ and need $t_i$ to advance, while all $p'_i$ are holding $t'_i$ and need $t'_{i+1}$ to advance. As $t_j \in \set{t'_1,\ldots, t'_n}$, all those processes are blocked, and in particular $t = t_0$ is held by a process which will never release it.
	
	As $p$ has a reachable transition taking $t$, we can define $\run_p$ as a shortest run that ends in a state where some outgoing transitions takes a lock of $\set{t_0, \ldots, t_m, t'_1, \ldots, t'_n}$. By minimality this run can  be executed, as all other locks are free. 
	By exclusiveness, it reaches a state where all transitions take the same non-free lock.
	
	By Lemma~\ref{lem:fair-exclusive} we can extend this run into a "process-fair" one, whose projection on $p$ can only be $\run_p$, as $p$ will never be able to advance further.
	
	\paragraph{Case 2: }Now suppose there is no $j$ such that $\run'_j$ has a weak pattern, then as we took all $\run'_j$ with weak patterns whenever possible, it means there is no local run with a weak pattern "witnessing" any of the $t'_j \edge{p'_j} t'_{j+1}$.
	We can then apply Lemma~\ref{lem:always-one-weak} to show that the reverse cycle $t'_1 = t'_{n+1} \edge{p'_n} t'_n \cdots \edge{p'_1} t'_1$ exists in $G$ and all its edges are "witnessed" by runs with weak patterns.
	Hence we can apply the arguments from the previous case using this cycle to conclude.
\end{proof}

\begin{lemma}
	\label{lem:cond-inf}
	If $p$ has a reachable transition acquiring some lock $t$ and there is a path in $G$ from $t$ to some $t'$ such that there is a process $q$ with an infinite local run $\run_q$ acquiring $t'$ and never releasing it (i.e., such that $t' \in \Owns(\run_q)$), then there is a "process-fair" global run with a finite projection on $p$. 
\end{lemma}
\begin{proof}
	Let $t = t_0 \edge{p_1} t_1 \edge{p_2} \cdots \edge{p_{k}} t_k = t'$ be the shortest path from $t$ to $t'$.
	Let $\run_p$ be a local run of $p$ acquiring $t$ at some point, either infinite or leading to a state with no outgoing transition. 
	For each $1 \leq i \leq k$ we select a local run $\run_i$ of $p_i$ "witnessing" $t_{i-1} \edge{p_i} t_i$.
	Furthermore we select those $\run_i$ with "weak patterns" whenever possible. 
	Let $t_q$ be the other lock used by $q$ besides $t'$, and let $\run_q$ be an infinite run of $q$ in which $t'$ is eventually taken and never released.
	We can decompose $\run_q$ as $u_q v_q$ where $u_q$ is the largest "neutral" prefix of $\run_q$. 
	We distinguish several cases:
	
	\paragraph{Case 1:} $t_q \notin \set{t_0, \ldots, t_k}$, or $t_q$ is not used in $v_q$. Then we can execute $u_q$, leaving all locks free, then $\run_1 \cdots \run_k$, which can be done as the execution of $\run_1 \cdots \run_i$ leaves $t_{i}, \ldots, t_k$ free and thus $\run_{i+1}$ can be executed.
	Let $v_q = v'_q v''_q$ with $v'_q$ a prefix of $v_q$ large enough so that $t'$ is held by $q$ and never released later.
	Then as no $t_i$ is used in $v_q$, we can execute $v'_q$. 
	Let $\run$ be the run constructed so far.
	Then by Lemma~\ref{lem:fair-exclusive} we can construct a "process-fair" run $\run'$ starting in the last configuration of $\run$ whose projection on $q$ is a prefix of $v''_q$ (thus $t_k$ is never released and thus neither are $t_0, \ldots, t_{k-1}$) and whose projection on $p$ is a prefix of $\run_p$ (and thus finite as $\run_p$ tries to acquire $t$, which is never free). As a consequence, $\run\run'$ is a "process-fair" run whose projection on $p$ is finite.
	
	\paragraph{Case 2:} $t_q = t_j$ for some $0 \leq j \leq k$ and $\run_q$ acquires $t_j$ at some point and never releases it. Then we apply the same reasoning as in the previous case for the path $t = t_0 \edge{p_1} \cdots \edge{p_j} t_j$.
	
	\paragraph{Case 3:} $t_q = t_j$ for some $0 \leq j \leq k$ and $t_j$ is used in $v_q$ but not kept indefinitely.
	
	\paragraph{Subcase 3.1:} there is an edge $t' \edge{q} t_j$.
	Then we have a path from $t$ to a cycle $t_j \edge{p_{j+1}} \cdots \edge{p_k} t' \edge{q} t_j$.
	Hence by Lemma~\ref{lem:cond-cycle}, there is a "process-fair" global run with a finite projection on $p$. 
	
	\paragraph{Subcase 3.2:} One of the runs $\run_i$ has a "weak pattern" $\dashover{\set{t_{i-1}}}$. 
	We decompose $\run_i$ as $u_i v_i$ with $u_i$ its largest "neutral" prefix.
	Then we execute $u_i$, then $\run_{i+1} \cdots \run_{k}$. After that we execute a prefix $\run'_q$ of $\run_q$ such that at the end $q$ holds only $t'$, and does not release it later. This prefix exists as $q$ never keeps $t_j$ indefinitely in $\run_q$. We decompose $\run_q$ as $\run'_q \run''_q$. Then we execute $\run_1 \cdots \run_{i-1} \run'_i$. All those runs can be executed as before executing each $\run_{i'}$ both locks of $p_{i'}$ are free, and before executing $v_i$, $t_{i-1}$ is free, which is all that is needed to execute $v_i$ as $\run_i$ has a "weak pattern". Let $\run$ be the run constructed so far.
	Then by Lemma~\ref{lem:fair-exclusive} we can construct a "process-fair" run from the configuration reached by $\run$ whose projection on $q$ is a prefix of $\run''_q$ and whose projection on $p$ is a prefix of $\run_p$.
	As a consequence, $t' = t_k$ is never released in $\run''_q$ and thus neither are $t_0, \ldots, t_{k-1}$. As $\run_p$ tries to take $t = t_0$ at some point, its prefix executed in $\run'$ is finite. 
	Hence $\run \run'$ is a "process-fair" run with a finite projection on $p$.
	
	\paragraph{Subcase 3.3:} There is no edge $t' \edge{q} t_j$ and all $\run_i$ have "strong patterns".
	When executing the $v_q$ part of $\run_q$, $q$ holds a lock at all times, and holds $t_j$ at some point and $t'$ at some point, hence it has to have both at the same time at some moment.
	Hence there is a moment at which $q$ holds one of the locks and is about to get the other. As the system is "exclusive", it means all its available transitions take that lock.
	Hence there is an edge  $t' \edge{q} t_j$ or $t_j \edge{q} t'$ in the graph. As we assumed that there is no edge $t' \edge{q} t_j$, there is one $t_j \edge{q} t'$. 
	Furthermore, as we selected the $\run_i$ so that they had weak patterns whenever possible, it means that for all $i$ there is no run with a weak pattern "witnessing" $t_{i-1} \edge{p_i} t_i$.
	By Lemma~\ref{lem:always-one-weak} this means that there are edges $t_k \edge{p_k} t_{k-1} \edge{p_{k-1}} \cdots \edge{p_{j+1}} t_j$. With the edge $t' \edge{q} t_j$, we obtain a cycle in $G$ with a path from $t$ to it.
	By Lemma~\ref{lem:cond-cycle}, there is a "process-fair" global run with a finite projection on $p$. 
	
	This concludes our case distinction, proving the lemma.	
\end{proof}

\begin{lemma}
	\label{lem:cond-stop}
	If $p$ has a reachable transition acquiring some lock $t$ and there is a path in $G$ from $t$ to some $t'$ such that there is a process $q$ with a local run $\run_q$ with $t' \in \Owns(\run_q)$ and going to a state with no outgoing transitions, then there is a "process-fair" global run with a finite projection on $p$. 
\end{lemma}
\begin{proof}
	Let $s_q$ be the state reached by $\run_q$, we add a self-loop on it with a fresh letter $\#$. As there are no other outgoing transitions from $s_q$ this does not break the exclusiveness. It does not change $G$ either.
	Then $\run_q \#^\omega$ is an infinite run acquiring $t'$ and never releasing it.
	
	Hence by Lemma~\ref{lem:cond-inf}, there is a "process-fair" run $\run$ in this new system whose projection on $p$ is finite. Let $h$ be the morphism such that $h(\#) = \epsilon$ and $h(a) = a$ for all other letters $a$. 
	Then $h(\run)$ is a "process-fair" run of the original system: it is a run as $\#$ does not change the configuration, meaning that all actions of $h(\run)$ can be executed. For the same reason, if a process $p'$ other than $q$ only has finitely many actions in $h(\run)$, then the same is true in $\run$, thus there is a point after which no configuration allows $p'$ to move in $\run$, and thus in $h(\run)$ as well. As for $q$, either it only executes $\#$ from some point on, meaning it has reached $s_q$ and will be immobilised in $h(\run)$, or it never executes any $\#$, in which case $h(\run) = \run$ and it follows the same configurations in both.
\end{proof}

\begin{lemma}\label{lem:exclusiveGraph}
	There is a "process-fair" run whose projection on $p$ is finite if and only if there is a local run $\run_p$ of $p$ leading to a state where all outgoing transitions take some lock $t$ and either 
	
	\begin{enumerate}		
		\item\label{Cx1} $p$ has a local run leading to a state with no outgoing transitions.
		
		\item\label{Cx2} or there is a path from $t$ to a cycle in $G$

		\item\label{Cx3} or there is a path in $G$ from $t$ to some lock $t'$ and there is a process $q$ with a local run $\run_q$ with an 
		"infinitary pattern" with $t' \in \Owns(\run_q)$.

		\item\label{Cx4} or there is a path in $G$ from $t$ to some lock $t'$ and there is a process $q$ with a local run $\run_q$ such that $t' \in \Owns(\run_q)$ and leading to a state with no outgoing transitions.
	\end{enumerate}
\end{lemma}

\begin{proof}
	We start with the left-to-right implication: Say there is a run $\run$ whose projection on $p$ is finite. For each process $p' \in \Proc$ let $\run_{p'}$ be its local run. 
	
	Then $\run_p$ has to end in a state where all available transitions acquire a
	lock $t$. If there are no transitions at all, condition~\ref{Cx1} is
	satisfied. If there is at least one such transition, then $t$ is held
	forever by some other process $p_1$.
	
	We construct a path $t = t_0 \edge{p_0} t_1 \edge{p_1} \cdots$ in $G$ so that all $t_i$ are held indefinitely by some process after some point in the run.
	Say we already constructed those up to $i$.
	
	There is a process $p_{i}$ holding $t_i$ indefinitely.
	If $\run_{p_{i}}$ is infinite, then condition~\ref{Cx3} is satisfied. Otherwise, $\run_{p_{i}}$ is finite, and with a "finitary pattern" such that $t_i \in \Owns(\run_{p_i})$.
	
	If this local run ends up in a state with no outgoing transition then condition~\ref{Cx2} is satisfied, otherwise it must have no choice but to acquire some lock $t'_{i+1}$.
	Hence we construct an infinite path
	$t'_0 \edge{p'_0} t'_1 \edge{p'_1} \cdots$ in $G$.
	
	The set of processes is finite, hence there exist 
	$i < j$ such that $t_i = t_j$, meaning we have reached a cycle. Thus condition~\ref{Cx4} is satisfied.
	
	For the other direction, suppose there exists $t$ as in the statement of the lemma, so that one of the conditions is satisfied.
	
	If condition~\ref{Cx1} is satisfied, then we have a finite run $\run_p$ leading to a state with no outgoing transition. We execute it and then prolong it into a global "process-fair" run by choosing a process uniformly at random and executing one of its available actions if there is any (similarly to the proof of \ref{lem:fair-exclusive}). 
	We obtain a "process-fair" run in which $p$ only has finitely many actions. 
	If condition~\ref{Cx2} is satisfied then we have the result by Lemma~\ref{lem:cond-cycle}.
	If condition~\ref{Cx3} is satisfied then we have the result by Lemma~\ref{lem:cond-stop}.
	If condition~\ref{Cx4} is satisfied then we have the result by Lemma~\ref{lem:cond-inf}.
\end{proof}

To conclude the proof of Proposition~\ref{prop:ptime-exclusive}, by Lemma~\ref{lem:exclusiveGraph}, we only have to check the four conditions listed in its statement.	
Here is our algorithm: 

We start by looking, in the transition system of process $p$, for a reachable local state with no outgoing transition. If there is one, we accept.

Then we compute all pairs $(q, t)$ such that either there is an infinite run of process $q$ keeping $t$ indefinitely from some point on or there is a run $\run_q$ with $t \in \Owns(\run_q)$ leading to a state with no outgoing transitions. 
As our system is "sound", the set of locks a process has is determined by its state. Let $\ownfunc_p$ be the function described in Definition~\ref{def:soundness}.
Then we compute all pairs $(t, t')$ such that some process $p'$ has a reachable state $s$ with $\ownfunc_{p'}(s) = \set{t}$ and all outgoing transitions of $s$ acquiring $t'$.
We obtain the edges of $G$.

For both locks of $p'$, we check that there is a reachable transition acquiring it, and there is a path in $G$ to either a cycle or to a $t$ from one of the pairs $(q,t)$ computed above. If it is the case for one of them, we accept, otherwise we reject.
This can all be done in polynomial time, proving the proposition.

\begin{restatable}{proposition}{PTIMEEXCLUSIVE}
	\label{prop:ptime-exclusive}
	The "process deadlock problem" is in \PTIME\ for "sound" "exclusive" "2LSS". 
\end{restatable}


\subsection{\NP-hardness for general 2LSS}

By contrast, when we lift the "exclusive" requirement, the problem becomes \NP-hard (and \NP-complete, as we will see later).

\begin{proposition}
	\label{prop:NPhard2LSS}
	The "process deadlock problem" is \NP-hard for "sound" "2LSS".
\end{proposition}

\begin{proof}
	
	\begin{figure}
		\begin{tikzpicture}[node distance=1.8cm,auto,>= triangle
	45,scale=.6]
	\tikzstyle{initial}= [initial by arrow,initial text=,initial
	distance=.7cm, initial where= left]
	\tikzstyle{accepting}= [accepting by arrow,accepting text=,accepting
	distance=.7cm,accepting where =right]
	
	\node[state,initial, minimum size=15pt] (p1) at (0,0) {};
	\node[state, minimum size=15pt] (p2) at (3,0) {};
	
	\node[state,initial, minimum size=15pt] (p'1) at (6,0) {1};
	\node[state, minimum size=15pt] (p'2) at (9,0) {2};
	\node[state, minimum size=15pt] (p'3) at (12,0) {3};

	\node[state, minimum size=15pt, initial] (Ci1) at (0,-3) {};
	\node[state, minimum size=15pt] (Ci2) at (3,-3) {};
	\node[state, minimum size=15pt] (Ci3) at (6,-3) {};
	
	\node[state, minimum size=15pt, initial] (li1) at (0,-6) {};
	\node[state, minimum size=15pt] (li2) at (3,-6) {};
	\node[state, minimum size=15pt] (li3) at (6,-6) {};

	\node[state, minimum size=15pt, initial] (xi1) at (9,-4) {};
	\node[state, minimum size=15pt] (xi2) at (12,-3) {};
	\node[state, minimum size=15pt] (xi3) at (12,-5) {};
	
	\node (x) at (8.4,-5) {\Large\textbf{$p^x_k$}};
	\node (l) at (-0.6,-7) {\Large\textbf{$p^\ell_{i,j}$}};
	\node (C) at (-0.6,-4) {\Large\textbf{$p^C_i$}};
	\node (p') at (5.4,-1) {\Large\textbf{$p'$}};
	\node (p) at (-0.6,-1) {\Large\textbf{$p$}};
	
	\path[->] 	
	(p1) edge node[above] {$\get{t}$} (p2)
	(Ci1) edge node[above] {$\get{t(C_i)}$} (Ci2)
	(Ci2) edge node[above] {$\get{t'}$} (Ci3)

	(xi1) edge node[below] {$\get{t(x_k)}$} (xi3)
	(xi1) edge node[above=10pt] {$\get{t(\neg x_k)}$} (xi2)
	;
	\path[->, bend right=20] 
	(p'1) edge node[below] {$\get{t}$} (p'2)
	(p'2) edge node[above] {$\rel{t}$} (p'1)
	(p'2) edge node[below] {$\get{t'}$} (p'3)
	(p'3) edge node[above] {$\rel{t'}$} (p'2)
	
	(li1) edge node[below] {$\get{t(C_i)}$} (li2)
	(li2) edge node[above] {$\rel{t(C_i)}$} (li1)
	(li2) edge node[below] {$\get{t(\ell_{i,j})}$} (li3)
	(li3) edge node[above] {$\rel{t(\ell_{i,j})}$} (li2)
	;
	\path[->, loop above] (p2) edge (p2);
\end{tikzpicture}
		\caption{Processes for the reduction in Proposition~\ref{prop:NPhard2LSS}}
		\label{fig:NPhard2LSS}
	\end{figure}
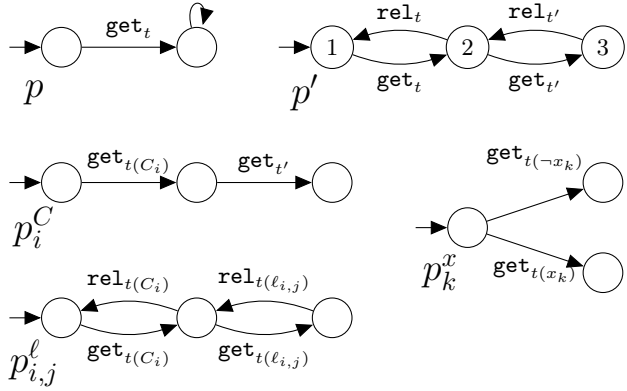
	
	We reduce from the 3SAT problem. 
	We use a set of variables $x_1, \ldots, x_n$.
	Let $\phi = \bigwedge_{i=1}^m C_i$ with for each $i$, $C_i = \ell_i^1 \lor \ell_i^2 \lor \ell_i^3$ with $\ell_i^j \in \set{x_k, \neg x_k \mid 1 \leq k \leq n}$.
	
	We construct a system with processes $\Proc = \set{p,p'} \cup \set{p^C_i \mid 1 \leq i \leq m} \cup \set{p^\ell_{i,j} \mid 1 \leq i \leq m, 1\leq j \leq 3} \cup \set{p^x_k \mid 1 \leq k \leq n}$.
	We also use locks $T = \set{t,t'} \cup \set{t(C_i) \mid 1 \leq i \leq m} \cup \set{t(x_k), t(\neg x_k) \mid 1 \leq k \leq n}$. 
	The transition systems of these process are described in Figure~\ref{fig:NPhard2LSS}. 
	
	In order to block process $p$ we need to block it in its first state by having another process keep $t$ forever.
	As a matter of fact, the only other process accessing $t$ is $p'$.
	As a consequence, a "process-fair" run blocks $p$ if and only if $p'$ eventually keeps $t$ forever.
	
	Consider such a run $\run$. Then eventually $p'$ has to stop visiting its
	state $1$. Furthermore, as $\run$ is "process-fair", $p'$ can never stay
	indefinitely in one of the other two states as it is always possible to
	execute a $\rel{}$ action. Hence $p'$ goes through states $2$ and $3$
	infinitely many times, meaning it takes and releases $t'$ infinitely often.
	
	This implies that none of the $p^C_i$ keep $t'$ indefinitely, which is only possible if all the $t(C_i)$ are taken and never released by other processes (if some $t(C_i)$ is free infinitely often, as $\run$ is "process-fair" $p^C_i$ has to take $t(C_i)$ at some point, and then $t'$ cannot be free infinitely often as $p^C_i$ would have to take it eventually).
	
	As a consequence, for each $C_i$ there has to be a $\ell_i^j$ such that $p_{i,j}^\ell$ keeps $t(C_i)$ forever, which is only possible if $t(\ell_{i,j})$ is free infinitely often.
	
	This means that the process $p^x_k$ (with $x_k$ the variable appearing in $\ell_{i,j}$) must have taken the lock associated with the negation of $\ell_i^j$ (it cannot stay in its initial state as the run is "process-fair" and $\ell_i^j$ is free infinitely often). 
	
	In conclusion, exactly one of $t(x_k), t(\neg x_k)$ is free infinitely often for each $k$, and for each clause $C_i$ there is a literal in $C_i$ whose lock is free infinitely often.
	Thus the valuation mapping each $x_k$ to $\top$ if $x_k$ is free infinitely often and $\bot$ otherwise satisfies $\phi$.
	
	Now suppose $\phi$ is satisfied by some valuation $\nu$.
	We construct the following run: 
	First of all for all $k$ process $p^x_k$ takes $t(\neg x_k)$ if $\nu(x_k) = \top$ and $t(x_k)$ otherwise.
	Then for each $i$ we select some $j_i$ such that $\nu$ satisfies $\ell_i^{j}$ and have process $p^\ell_{i,j_i}$ take $C_i$.
	Finally, $p'$ takes $t$.
	
	We then repeat the following steps indefinitely: one by one each $p^\ell_{i,j_i}$ takes $t(\ell_{i}^{j})$ and releases it, then $p'$ takes and releases $t'$.
	This is all possible as all $t(\ell_{i}^{j_i})$ are free (those $\ell_i^j$ are satisfied by $\nu$ hence the corresponding $p^x_k$ took their negations) and so is $t'$ (none of the $p^{C}_i$ ever moves thus they do not take $t'$).
	
	This run is "process-fair" as the processes that are eventually blocked are the $p^x_k$ (which end up in states with no outgoing transitions), the $p^C_i$ (which need $C_i$ to advance, but those locks are never released) and $p$ (which needs $t$ to move on, but $t$ is kept forever by $p'$).
	This concludes our reduction.
\end{proof}

\section{Regular objectives}
\label{sec:reg-objectives}

\subsection{The problem is \PSPACE-complete in general}

In order to justify our approach, we prove that the general verification of "LSS" against "regular objectives" is \PSPACE-complete, even with strong restrictions on the system.

\begin{restatable}{proposition}{PSPACELSS}
	\label{prop:PSPACE-LSS}
	The "regular verification problem" is \PSPACE-complete for "LSS" in general.
	\PSPACE-hardness already holds for the "process deadlock problem" for "sound" "exclusive" "LSS" even with a fixed number of locks per process.
\end{restatable}

	The \PSPACE~upper bound is easy to obtain: It suffices to guess a state $s_p$ in each $\aut_p$ and $s'_p$ in each $\objaut{p}$, and then guess a sequence of letters in $\S$ while keeping track of the states reached by that sequence in the $\aut_p$ and $\objaut{p}$. 

If we reach a configuration with each $\aut_p$ in state $s_p$ and each $\objaut{p}$ in $s'_p$, we start memorising the set of visited states in each $\objaut{p}$. 
If we reach that configuration again, we stop and accept if and only if the set of visited states in the $\objaut{p}$ satisfies $\phi$.
This comes down to guessing an ultimately periodic run in the global system and checking that it satisfies the objective.

The difficulty is to obtain the \PSPACE-hardness with a fixed number of locks per process. To do so we reduce the emptiness problem for the intersection of a set of deterministic automata. 

Without loss of generality we will assume that there are at least two automata, that they are all over alphabet $\set{0,1}$, and that their languages are all included in $11(00+01)^*$: we can always apply a small transformation to each automaton so that, if its language was $\mathcal{L}$, it becomes $11 h(\mathcal{L})$ with $h$ the morphism mapping $0$ to $00$ and $1$ to $01$. The intersection of those languages is empty if and only if the intersection of the original languages was empty.

Let $\aut_1, \cdots, \aut_n$ (with $n \geq 2$) be automata, with, for each $1 \leq i \leq n$, $\aut_i = (S_i, \set{0,1}, \delta_i, init_i, F_i)$.
We construct a "sound" "exclusive" "LSS" $\tss$ as follows:

For each $1 \leq i \leq n$ we have a process $p_i$ which is in charge of simulating $\aut_i$.	
The set of locks is $T = \set{0_i, 1_i, key_i \mid 1 \leq i \leq n}$.
For all $i$, $p_i$ accesses locks $0_i, 1_i, key_i$, as well as $0_{i+1}, 1_{i+1}, key_{i+1}$ if $i \leq n-1$ and $0_{1}, 1_1$ if $i =n$. Thus a process uses at most 6 locks in total.

For all $1 \leq i \leq n$ and $t$ accessed by $p_i$, we have two actions $\get{t}^i$ and $\rel{t}^i$, with which $p_i$ acquires and releases lock $t$, as well as actions $\nop^i$ and $end^i$ with no effect on locks.

In the proof the following local sequences will be important:

\[\send(0) = \rel{0_n}^n \get{0_{0}}^n \rel{1_n}^n \get{1_0}^n \rel{0_{0}}^n \get{0_n}^n \rel{1_{0}}^n \get{1_n}^n\]

\[\rec_i(0) = \nop^i \get{0_{i+1}}^i \rel{0_{i}}^i \get{1_{i+1}}^i \rel{1_{i}}^i \get{0_{i}}^i \rel{0_{i+1}}^i \get{1_{i}}^i \rel{1_{i+1}}^i\]

$\send(1)$ and $\rec_i(1)$ are defined analogously, by replacing $0$ by $1$ and $1$ by $0$ everywhere.

The following global sequences will be useful as well:

\[ \acqs(0) = \get{0_{n}}^{n-1}\rel{0_{n-1}}^{n-1} \cdots  \get{0_{1}}^{1}\rel{0_{0}}^{1}\]

\[ \rels(0) = \get{0_{0}}^{1}\rel{0_{1}}^{1} \cdots  \get{0_{n-1}}^{n-1}\rel{0_{n}}^{n-1}\]

\[ \NOP = \nop^{1} \cdots  \nop^{n-1}\]

$\acqs(1)$ and $\rels(1)$ are defined analogously, by replacing $0$ by $1$ and $1$ by $0$ everywhere.	

The transition system of each process $p_i$ is designed as follows: 
We start with $\aut_i$, and we replace every transition labelled $0$ with a sequence of transitions labelled by actions of $\send(0)$ if $i=n$, and $\rec_i(0)$ if $1 \leq i \leq n-1$ (there is at least one such $i$ as $n \geq 2$).

Furthermore we add a few transitions so that each $p_i$ with $i \leq n-1$ executes $\start_i = \get{key_{i+1}}^i \get{0_i}^i \get{1_i}^i  \get{key_i}^i \rel{key_{i+1}}^i$ before entering the initial state of $\aut_i$.
If $i=n$ that sequence is $\start_n = \get{a_n}^n \get{b_n}^n \get{key_n}^n$.
We also add a transition reading $end^i$ from all states of $F_i$ to a state $stop_i$ with no outgoing transition.

The objective is that the action $end^i$ is executed for all $1 \leq i \leq n$.

One direction is easy.
Say there is a word $u = b_1 b_2 \cdots b_m$ in the intersection of the languages of the $\aut_i$. Then we start by executing all $\start_i$ sequences for all $i$ in increasing order, and then, for each $1 \leq j \leq m$ (in increasing order), we execute the sequence of operations

\begin{align*}
	seq(b_j) = &\rel{(b_j)_n}^n \NOP~\acqs(b_j) \get{(b_j)_0}^n \rel{(1-b_j)_n}^n\\ &\acqs(1-b_j) \get{(1-b_j)_0}^n \rel{(b_j)_0}^n \rels(b_j)\\ &\get{(b_j)_n}^n\rel{(1-b_j)_0}^n \rels(1-b_j) \get{(1-b_j)_n}^n 
\end{align*}

This run projects on $p_i$ as $\start_i \rec_i(b_1) \rec_i(b_2) \cdots \rec_(b_m)$ if $i\leq n-1$ and $\start_i \send(b_1) \send(b_2) \cdots \send(b_m)$ if $i=n$. As $u$ is in the language of $\aut_i$, $p_i$ can execute this run locally.
It can be easily checked that all operations in that run are valid in the current configuration, hence this sequence can be executed.

As $u$ is accepted by all $\aut_i$, after executing the sequence above each process $p_i$ ends up in a state of $F_i$, and thus they can all execute $end^i$ one after the other.

Conversely, suppose there is some run $\run$ whose local projection $\run|_p$ on each process $p$ is ends with $end^i$.
Each $\run|_{p_i}$ must start with the execution of $\start_i$.

We prove the following lemma:

\begin{lemma}
	For all $j \in \nats$, let $\run_j$ be the shortest prefix of $\run$ whose projection on $p_n$ is $\start_n \send(b_1) \cdots \send(b_j)$. 
	Then the projection of $\run_j$ on every other $p_i$ has $\rec_i(b_1) \cdots \rec_i(b_j)$ as a suffix.
\end{lemma}

\begin{proof}
	We prove this by induction on $j$.
	For $j=0$ it is trivial.
	Let $j \in \nats$, suppose the claim is true for $j$, we show it for $j+1$.
	
	First of all note that for all $i \geq 1$, if $p_i$ has finished executing $\start_i$ then it holds $k_i$ and will never release it, hence $p_{i-1}$ either has executed $\start_{i-1}$ in full or has not begun executing it. In the second case, $p_{i-1}$ will never be able to advance, which is impossible as $\run|_p$ is not empty. As a result, after $p_n$ has executed $\start_n$, all other $p_i$ must have executed $\start_i$.
	
	Another important remark is that after executing $\start_i$, all $p_i$ alternate between a $\get{}$ and a $\rel{}$ no matter which local run they execute. While $p_n$ starts with a $\rel{}$, all other processes start with a $\get{}$. Therefore $p_n$ always holds either 2 or 3 locks, while all others always hold either 3 or 4.
	There are $n$ processes and $3n$ locks in total, hence at all times the global configuration is such that either all locks are taken and the next operation of some process $p$ is $\rel{}$ (and $\get{}$ for all others) or one lock is free and all processes have a $\get{}$ as their next operation.
	
	Now say process $p_n$ has started executing $\send(b_{j+1})$ by releasing $(b_{j+1})_n$. This means some other process must have taken $(b_{j+1})_n$, which can only be $p_{n-1}$. The only possibility is that $p_{n-1}$ then releases $(b_{j+1})_{n-1}$, which can only be taken by $p_{n-2}$, ... We must end up executing $\acqs(b_{j+1})$, which ends with $(b_{j+1})_0$ free, which can only be taken by $p_n$.
	
	By continuing this reasoning we conclude that $\run_{j+1}=\run_j seq(b_{j+1})$, proving the lemma.
\end{proof}

As all $p_i$ execute $end_i$ in $\run$, in particular $\run|_{p_n}$ ends with $end^n$, hence it is necessarily of the form $\start_n \send(b_1) \cdots \send(b_m) end^n$. 
Let $\run'$ be $\run$ where all $end^i$ have been erased.
The lemma above allows us to conclude that for all $i$, $\run'|_{p_i}$ has $\rec_i(b_1)  \cdots \rec_i(b_m)$ as a suffix. 

Recall that the languages of all $\aut_i$ are included in $11(00+01)^*$.
Moreover, we know that $b_1 \cdots b_m$ is in the language of $\aut_{n}$,
hence $b_1 = b_2 =1$. Furthermore, all $\run'|_{p_i}$ ($i<n$) are of the form $\start_i \rec_i(x_1) \cdots \rec_i(x_r)$ with $x_1 \cdots x_r$ in the language of $\aut_i$.
As the only moment a factor $11$ can appear in a word of those languages is at the beginning, for $\run'|_{p_i}$ to have $\rec_i(b_1) \cdots \rec_i(b_m)$ as a suffix, we must have $\run'|_{p_i} = \start_i \rec_i(b_1) \cdots \rec_i(b_m)$. 

As a consequence, for all $i$ we have $\run|_{p_i} = \start_i \rec_i(b_1) \cdots \rec_i(b_m) end^i$ and thus $b_1 \cdots b_m$ must be accepted by all $\aut_i$.

We have proven that this system had a run in which each $p_i$ reads $end^i$ if and only if there is a word accepted by all $\aut_i$.

As that condition is easily expressible as a "regular objective", we obtain the \PSPACE-hardness of the "regular verification problem" for "sound" "exclusive" "LSS" with 6 locks per process.
However, our goal was to prove the \PSPACE-hardness of the "process deadlock problem" for "sound" "exclusive" "LSS".

To do so, we add a process $q$ and locks $\ell_{i}$ for $1 \leq i \leq n+1$ so that the transition system of $q$ simply takes $\ell_{n}$ and then goes to a state with a self-loop executing $\nop$. We also add, for each $1 \leq i \leq n$, a sequence of transitions from $stop_i$ which take $\ell_i$, then $\ell_{i-1}$ and release $\ell_i$ if $i \geq 2$, and simply take $\ell_i$ if $i=1$, to end up in a state with no outgoing transition.

We show that there is a "process-fair" run with a finite projection on $q$ if and only if there is one in the previous "LSS" such that each $p_i$ executes $end_i$.

If the latter is true, then we just take the same run and prolong it so that each $p_i$ takes $\ell_i$. Then all $\ell_i$ are taken, and all processes need some $\ell_i$ to advance, we have reached a global deadlock (in particular the run is "process-fair", and its projection on $q$ is finite).

Conversely, suppose we have a "process-fair" run in the new "LSS" with a finite projection on $q$. Then $q$ must be blocked, which is only possible if $\ell_n$ is held forever by $p_n$, which in turn is only possible if $\ell_{n-1}$ is held forever by $p_{n-1}$... We conclude that all $p_i$ must be holding $\ell_i$ forever from some point on, and thus that they all read $end_i$.

We project that run to erase all actions getting an $\ell_i$. We obtain a run of the previous system in which every process has executed $end_i$.

As a result, the new "LSS" has a "process-fair" run in which $q$ is blocked if and only if the former "LSS" has a run in which every $p_i$ has executed $end_i$, if and only if the $\aut_i$ recognise a common word.

As a result, the "process deadlock problem" is \PSPACE-complete for "sound" "exclusive" "LSS" with 8 locks per process.


\subsection{...but \NP-complete for 2LSS}

Then we prove that the complexity falls to \NP\ when we demand that each process
uses at most two locks.

\begin{proposition}
	\label{prop:NP-2LSS}
	The "regular verification problem" is \NP-complete for "2LSS" (the lower bound holds even for "sound" "exclusive" "2LSS").
\end{proposition}

\begin{proof}
	We start with the upper bound. 
	Let $\tss = ((\aut_p)_{p\in \Proc}, T, op)$ be a "2LSS" and $((\objaut{p})_{p\in\Proc},\phi)$ a "regular objective".
	Our \NP~algorithm goes as follows:
	we guess a pattern $\Pat_p$ for each process $p$, as well as a valuation $\nu$ of the $(\varinf_{p,s})_{p\in \Proc, s \in S_p}$.
	For each $p$ let $\Owns_p$ be the set of locks kept indefinitely by a run respecting $\Pat_p$.
	
	Then we check that those patterns respect the conditions of Proposition~\ref{prop:charac-schedulable-pat} and that this valuation satisfies $\phi$ (otherwise we stop).
	We then equip each $\objaut{p}$ with the acceptance condition
	$\bigwedge_{\nu(\varinf_{p,s}) = \top} \varinf_{s} \land
	\bigwedge_{\nu(\varinf_{p,s}) = \bot} \neg\varinf_{s}$
	
	We add a self-loop labelled $\dummy$ on each state in $\aut_p$ whose outgoing transitions all acquire a lock of $\bigcup_{p \in \Proc} \Owns_{p}$.
	
	Then, for each $p$ we construct the product $\mathcal{C}$ of $\aut_p$, $\objaut{p}$ and $\aut_{\Pat_p}$ (from Lemma~\ref{lem:DELA-pat}) to obtain an "ELA" recognising runs of $p$ that match pattern $\Pat_p$ and are in the language of $\objaut{p}$.
	We guess an ultimately periodic run of the form $uv^\omega$ with $u$ and $v$ of polynomial size in the number of states of $\mathcal{C}$ and check that it is accepting (otherwise we stop). It is well-known that an "ELA" either has an empty language or accepts a run of that form.
	Then we accept.	
	
	We accept if and only if there is a valuation $\nu$ satisfying $\phi$ and a family of patterns $(\Pat_p)_{p \in \Proc}$ such that there exist local runs $(\run_p)_{p \in \Proc}$ of the processes matching those patterns and producing words whose runs in the $(\objaut{p})_{p \in \Proc}$ match $\nu$, and such that the finite ones end in states from which they can only take locks of $\Owns_p$. 
	By Proposition~\ref{prop:charac-schedulable-pat}, this is true if and only if there is a global run of the system satisfying the given objective.
	Hence the problem is in \NP.
	
	For the lower bound, we could easily translate a SAT formula into a "regular objective", with one process for each variable choosing to set it to $\top$ or $\bot$.
	
	However, we want to show that the \NP~complexity lies already in the model with no need for complicated objectives.
	By Proposition~\ref{prop:NPhard2LSS}, we know that the existence of a "process-fair" run blocking a given process $p$ is \NP-hard for "sound" "2LSS".
	However this is not the case if we are restricted to "exclusive" "2LSS".
	
	In order to prove the lower bound for "exclusive" "2LSS" we adapt the reduction from the proof of Proposition~\ref{prop:NPhard2LSS}. 
	Note that the only non-exclusive processes in Figure~\ref{fig:NPhard2LSS} are $p'$ and $p^\ell_{i,j}$. In $p'$ we add an extra state $4$ and replace the transition from $2$ to $3$ with a $\nop$ transition from $2$ to $4$ and a $\get{t'}$ transition from $4$ to $3$.
	What may then happen is that $p'$ gets stuck in $4$ because $t'$ is taken by some other process forever, which could not happen before as $p'$ always had the option of releasing a lock in $2$. To overcome this, we add to the objective that $p'$ should have an infinite run.
	We do the same thing for $p^\ell_{i,j}$, by decomposing the $\get{t(\ell^j_i)}$ into two transitions and adding the requirement that all $p^\ell_{i,j}$ should run forever.
	The proof is then exactly the same as the one for Proposition~\ref{prop:NPhard2LSS}. 
\end{proof}

%
%
	


\section{Nested locks}
\label{sec:nest}

In this section we address the verification problem for systems with a restricted lock acquisition policy.
We require that each process acquires and releases locks as if they were stored in a stack.
This is a classical restriction, as this way of managing locks is considered to be sound and suitable in many contexts.

\AP An "LSS" is ""nested"" if all its runs are such that a process can only release the lock it acquired the latest among the ones it holds.
In~\cite{Brotherson21} (Theorem~5.5) the authors considered a type of system which can be translated to our "sound" "nested" "exclusive" "LSS" and proved an \NP~upper bound on the complexity of the following problem: Is there a reachable configuration where there are some processes $p_1, \ldots, p_k \in \Proc$ and locks $t_1, \ldots, t_{k+1}=t_1 \in T$ with each $p_i$ holding lock $t_i$ and needing to get $t_{i+1}$ to keep running? We will call such configurations ""circular deadlocks"". They leave the question of a matching lower bound open.

We considerably generalise their result by proving an \NP~upper bound on the "regular verification problem" for
"nested" "LSS" (note that the problem above can be solved by guessing a configuration with such a circular deadlock and using our \NP~algorithm to check reachability of that configuration).
We then prove an \NP~lower bound on the "process deadlock
problem" for "sound" "nested" "exclusive" "LSS", thereby adding a matching \NP~lower bound to their result.

This shows that the "nested" requirement significantly improves the complexity
of the "regular verification problem". On the other hand, the \NP-hardness is
difficult to avoid: it holds even for a very restricted class of systems and for
very simple objectives.

\begin{restatable}{lemma}{STAIR}
	\label{lem:stair-dec}
	Every local run in a "nested" LSS can be decomposed as 
	\[\run=\run_0 a_1 \run_1 a_2 \cdots \run_{k-1} a_k \run_{k} \run_{k+1} \cdots\]
	
	\noindent where $a_1, \ldots, a_k$ are the actions getting a lock that is not released later in $w$.
	
	Furthermore, all $\run_{i}$ are "neutral".
	Finally, for all $i \geq k+1$, all locks acquired in $\run_i$ are acquired infinitely many times in $\run$.
	If $\run$ is finite, all $\run_i$ are empty for $i \geq k+1$.
	We call this decomposition the \introinrestatable{stair decomposition} of $\run$.
\end{restatable}

\begin{proof}
	Let $\run$ be a local run of some process $p$. 
	We start by decomposing it as 
	\[\run=\run_0 a_1 \run_1 a_2 \cdots \run_{k-1} a_k \run_{\infty}\]
	with $a_1, \ldots, a_k$ the actions getting a lock that is not released later
	in the run. 
	For all $i$ let $t_i$ be the lock taken by $a_i$,  namely $op(a_i) = \get{t_i}$.
	
	We check that all $\run_0,\dots,\run_{k-1}$ are "neutral". Consider some $\run_i$. 
	If a lock $t$ is taken in $\run_i$ then it must be released later in the run
	because $a_{i+1}$ is the next operation that takes a lock and does not release
	it.
	But because of the nesting discipline $t$ cannot be released after $a_{i+1}$. 
	So it must be released in $\run_i$.
	
	
	
	Now we look at $\run_\infty$. Every lock acquired in it must be
	released eventually. 
	Thus if the run is finite we can set $\run_{k} =\run_{\infty}$ and $\run_i=\epsilon$ for all $i \geq k+1$.
	
	If the run is infinite then we proceed as follows: 
	Before executing $\run_\infty$, $p$ holds $t_1, \ldots, t_k$.
	We construct a sequence of "neutral" runs $\run'_j$ such that $\run_\infty = \run'_{1} \run'_2 \cdots$.
	Say we constructed $\run'_1 \cdots \run'_j$.
	As they are all "neutral", after executing them $p$ holds $t_1, \ldots, t_k$.
	The next action $a$ in $\run_\infty$ cannot release a lock as none of those locks are ever released.
	If $a$ does not get a lock then we can simply set $\run_{j+1} = a$.
	If $a$ acquires lock $t$ then let $\run_{j+1}$ be the infix of $\run_\infty$
	starting with $a$ and ending with the next action releasing $t$. 
	This run is "neutral" as the system is "nested".
	Then let $j$ be such that $\run'_{1} \cdots \run'_{j}$ contains all $\get{t}$ operations with $t$ acquired finitely many times in $\run_{\infty}$.
	We set $\run_{k} = \run'_{1} \cdots \run'_{j}$ and for all $i \geq k+1$, $\run_i = \run'_{j-k+i}$. We obtain our decomposition.
\end{proof}


%

We now define patterns of local runs in a similar manner as in Section~\ref{sec:patterns}.	

\begin{definition}\label{def:nested-pattrerns}
	Consider a (finite or infinite) local run $\run$ of process $p$, and its "stair decomposition" $\run=\run_0 a_1
	\cdots \run_{k-1}   a_k \run_{k} \run_{k+1} \cdots$. For all $i$ let $t_i$ be the lock acquired by $a_i$.
	
	We say that $\run$ matches a \intro{stair
	pattern} $(\OwnsN(\run),\leq^{\run}, \InfN(\run))$ when
	$\OwnsN(\run)=\set{t_1,\dots,t_k}$, the set of locks acquired
	infinitely many times is included in $\InfN(\run)$, and
	 $\leq_p^{\run}$ is a total order
	on $T$  satisfying two conditions:
	\begin{itemize}
		\item if $t$ is acquired finitely many times and $t'$ infinitely many
		times then $t \leq^{\run}_p	t'$,
		\item if $t=t_i$ for some $i$
		and $t'$ is acquired at some point after $a_i$ then $t \leq^{\run}_p	t'$.
	\end{itemize}
\end{definition}
The $N$ in exponent above $\OwnsN_{p}$ and $\InfN_p$ is for nested, to avoid confusion with the notations defined in Section~\ref{sec:patterns}: while $\OwnsN$ and $\Owns$ correspond to the same idea, $\InfN$ and $\Inf$ are two different things. 

Note that unlike the patterns defined for "2LSS", here a run may have several different patterns. We could define unique patterns but this would somehow make the statement of Lemma~\ref{lem:pattern-char-nested} and the proof of Lemma~\ref{lem:DELA-pat-nested} more complicated.

Our next lemma characterises when local runs can be combined into a
"process-fair" global one.
Once again the characterisation uses only patterns and last states of the local runs.

\begin{lemma}
	\label{lem:pattern-char-nested}
	Consider a family of (finite or infinite) local runs $(\run_p)_{p \in \Proc}$ of a "nested" "LSS".
	For each $p \in \Proc$ we consider a "stair decomposition" of $\run_p$:
	
	\[\run_p=\run_{p,0} a_{p,1} \cdots \run_{p,k_p-1} a_{p,k_p} \run_{p,k_p} \run_{p,k_p+1} \cdots\]
	
	and for each $a_{p,i}$ let $t_{p,i}$ be the lock such that $op(a_{p,i}) = \get{t_{p,i}}$. 
	
	Runs $(\run_p)_{p \in \Proc}$  can  be scheduled into a "process-fair" global run if and only if there exist for each $p$ a "stair pattern" $(\OwnsN_p, \leq_p, \InfN_p)$ that $\run_p$ matches and  the following conditions are satisfied.
	
	\begin{enumerate}
		\item \label{C1n} The $\OwnsN_{p}$ sets are pairwise disjoint.
		
		\item \label{C2n} All $\leq_p$ orders are the same.
		
		\item \label{C3n} For all $p$, if $\run_{p}$ is finite then it leads to a state where all outgoing transitions acquire a lock from $\bigcup_{p \in \Proc} \OwnsN_{p}$. 
		
		\item \label{C4n} The set $\bigcup_{p \in \Proc}\OwnsN_p$ is disjoint from $\bigcup_{p \in \Proc}\InfN_p$. 
	\end{enumerate}
\end{lemma}

\begin{proof}
	Suppose we have a "process-fair" global run $\run$ whose local projections are
	the $(\run_p)_{p \in \Proc}$. For each $p$ let $\OwnsN_p$ be the set of locks kept indefinitely
	in $\run_p$ and $\InfN_p$ the set of locks acquired infinitely often in
	$\run_p$. 
	Let $\leq$ be a total order on locks such that for all $t, t' \in T$, if $t$ is acquired finitely many times in
	$\run$ and there is an operation on $t'$ after the last operation on $t$ then $t
	\leq t'$. 
	In particular, a lock acquired infinitely often is always greater than one
	acquired finitely many times. 
	Further, for all $p,i$ the action $a_{p,i}$ acquires $t_{p,i}$, which is not released later. Thus $a_{p,i}$ is the last action with an operation on $t_{p,i}$ in $\run$. Hence if another lock $t$ is used after $a_{p,i}$ in $\run_p$, it is also used after $a_{p,i}$ in
	$\run$, and therefore $t_{p,i} \leq t$. 
	As a result, $(\OwnsN_p, \leq, \InfN_p)$ is a pattern of $\run_p$ for all $p$, 
	and~\ref{C2n} is immediately satisfied. 
	
	As each $p$ eventually holds $\OwnsN_p$ and keeps those locks forever, the $\OwnsN_{p}$ have to be disjoint, thus condition~\ref{C1n} is satisfied. 
	
	For condition~\ref{C3n}, we use the fact that $\run$ is "process-fair". For
	all $p$, if $\run_p$ is finite then it leads to a state where after some point
	in the run none of  the outgoing transitions can be executed.
	Hence all these transitions acquire a lock that is never released after
	some point. 
	This is the case for locks of $\bigcup_{p \in \Proc} \OwnsN_p$ but not for the
	others, which are free infinitely often. Hence condition~\ref{C3n} holds. 
	
	Finally, as all locks from $\bigcup_{p \in \Proc}\OwnsN_p$ are eventually never free while the locks from $\bigcup_{p \in \Proc}\InfN_p$ are free infinitely often, the two sets are necessarily disjoint, proving condition~\ref{C4n}.
	
	For the other implication, suppose that we have patterns $(\OwnsN_{p}, \leq_p, \InfN_p)$ such that all conditions are satisfied. Let $\leq$ be the total order on locks common to all patterns, which exists by condition~\ref{C2n}.
	We start by executing one by one for each run $\run_p$ its prefix $\run_{p,0}$, leaving all locks free are the $\run_{p,0}$ are all "neutral".
	
	We use the notation $T_O$ for the set $\bigcup_{p\in\Proc} \OwnsN_{p}$.
	We index the locks of $T_O$ so that $T_O = \set{t_1, \ldots, t_m}$ and $t_1 \leq t_2 \leq \cdots \leq t_m$.
	For each $t_i \in T_O$ there is a pair $(p_i,j_i)$ such that $op(a_{p_i,j_i})
	= \get{t_i}$. Furthermore that pair is unique as a process $p$ cannot have
	$a_{p,j}$ take $t_i$ for two different $j$ (by definition of "stair
	decomposition") and as the $\OwnsN_{p}$ are disjoint (by condition~\ref{C1n}).  
	We execute, for all $t_i \in T_O$, in increasing order on $i$, $a_{p_i,j_i} \run_{p_i,j_i}$.
	
	At first all locks are free. Then, for each $i$, just
	before we execute $a_{p_i,j_i} \run_{p_i,j_i}$, the locks that are not free are exactly $\set{t_{i'} \mid i' \leq i-1}$. 
	Hence for every lock $t_{i'}$ that is not free, we have $t_{i'} \leq t_i$ and $t_{i'}
	\neq t_i$. 
	
	By definition of $\leq_p$, all locks $t'$ acquired in $a_{p_i,j_i} \run_{p_i,j_i}$ are such that $t_i \leq_p t'$, hence $t_i \leq t'$ by condition~\ref{C2n}. 
	As a result, they are all free just before we execute $a_{p_i,j_i} \run_{p_i,j_i}$.
	After we execute it, the set of non-free locks becomes $\set{t_{i'} \mid i' \leq i}$.
	
	The projection of the resulting run on each $p$ is $\run_{p,0} a_{p,1}  \cdots \run_{p,k_p-1} a_{p,k_p} \run_{p,k_p}$.
	
	All that is left to do is executing the $\run_{p, i}$ for $i \geq k_p +1$ for each $p$. They only contain operations on locks that are acquired infinitely many times which are thus in $\InfN_p$ as $\run_p$ matches pattern $(\OwnsN_{p}, \leq_p, \InfN_p)$, and therefore free by condition~\ref{C4n}. 
	As furthermore all $\run_{p,i}$ are "neutral" by definition of "stair decomposition", we can execute the next $\run_{p,i}$ for each $p$ again and again indefinitely, to obtain an infinite global run of the system.
	
	This run is furthermore "process-fair" as the finite $\run_p$ lead to states whose outgoing transitions acquire locks of $T_O$, which are eventually all taken forever.
	Hence those processes do not have an available action infinitely often.   
\end{proof}

Before we can present our \NP~algorithm, we need one last technical lemma to show that we can recognise runs with a given pattern using a small automaton.

\begin{restatable}{lemma}{DELANested}
	\label{lem:DELA-pat-nested}
	Given a process $p$ and a "stair pattern" $\Pat$ we can construct an "ELA"
	$\aut^p_{\Pat}$ such that for all nested local run $\run_p$, $\pad{\run_p}$ is accepted if and only
	if $\run_p$ matches "stair pattern" $\Pat$.
	The automaton $\aut^p_{\Pat}$ has at most $(\size{T_p}+2)^2$ states and a formula of
	constant size for the accepting condition.
\end{restatable}

\begin{proof}
	Let $\Pat = (\OwnsN_{p}, \leq_p, \InfN_p)$.
	We set $\OwnsN_{p} = \set{t_1, \ldots, t_k}$ so that $t_1 \leq_p \cdots \leq_p t_k$.
	
	We define the automaton $\aut^p_{\Pat} = (S^p_{\Pat}, \Sigma_p \cup \set{\dummy}, \Delta^p_{\Pat}, init^p_{\Pat}, \phi^p_{\Pat})$ as follows:
	If there exist $t, t'$ such that $t \in \InfN_p$, $t'\notin \InfN_p$ and $t \leq_p t'$ then no run can match this "stair pattern", hence we simply set $\aut^p_{\Pat}$ as an automaton with an empty language. From now on we will assume that it is not the case.
	
	The states of the automaton are $S^p_{\Pat} = \set{0,\ldots,k, \infty} \times (T_p \cup \set{neutral})$, with $init^p_{\Pat} = (0,neutral)$.
	
	\paragraph{Intuition} The first component of each state gives an index $i$ such that the run read so far is of the form $\run_0 a_1 \run_1 \cdots a_i \run_i$ with $\run_j$ "neutral" for all $j < i$, and for all $j \leq i$ $op(a_j) = \get{t_j}$ and all locks $t'$ used after $a_j$ are such that $t_j  \leq_p t'$. If the first component is $\infty$ it means we will only use locks of $\InfN_p$ in the future.
	
	The second component of a state $(i, x)$ indicates which lock apart from $\set{t_1, \ldots, t_i}$ we acquired earliest among the ones we own. If we released all locks acquired since we took $t_i$, then the second component is $neutral$.
	We do not need to keep track of all locks acquired as we are only interested in "nested" runs: If we are in state $(i,neutral)$ and acquire some lock $t$, we go to state $(i,t)$ to wait for it to be released: if we stay in state $(i,t)$ indefinitely the run is not accepted, otherwise $t$ is released we know that if the run we read is "nested" then all locks taken since we took $t$ have been released before.

	\paragraph{Formal proof:}
	For each action $a \in \Sigma_p \cup \set{\dummy}$ and state $s \in S^p_{\Pat}$ we have the following
	transitions:
	\begin{itemize}
		\item If $s=(i,neutral)$ with $i < k$ then:
		\begin{itemize}
			\item If $op(a) = \get{t_{i+1}}$ then $\Delta^p_{\Pat}(s,a) = \set{(i+1, neutral), (i,t_{i+1})}$
			
			\item If $op(a) = \get{t}$ with $t_i \leq_p t$ then  $\Delta^p_{\Pat}(s,a) = \set{(i,t)}$
		\end{itemize}

		\item If  $s=(k,neutral)$ then:
		\begin{itemize}
			\item If $op(a) = \get{t}$ with $t_k \leq t$ then $\Delta^p_{\Pat}(s,a) = \set{(k, t)}$
		\end{itemize}
		
		\item If  $s=(i,t)$ with $i \leq k$ then:
		\begin{itemize}
			\item If $op(a) = \rel{t}$ then $\Delta^p_{\Pat}(s,a) = \set{(i, neutral)}$
			
			\item If $op(a) = \get{t'}$ or $\rel{t'}$ with $t_i \leq_p t'$ and $t' \neq t_i$ then  $\Delta^p_{\Pat}(s,a) = \set{s}$
		\end{itemize}
		
		\item If  $s=(\infty,neutral)$ then:
		\begin{itemize}
			\item If $op(a) = \get{t}$ with $t\in \InfN_p$ then $\Delta^p_{\Pat}(s,a) = \set{(\infty, t)}$
		\end{itemize}
		
		\item If  $s=(\infty,t)$ then:
		\begin{itemize}
			\item If $op(a) = \rel{t}$ then $\Delta^p_{\Pat}(s,a) = \set{(\infty, neutral)}$
			
			\item If $op(a) = \get{t'}$ or $\rel{t'}$ with $t' \in \InfN_p$ then $\Delta^p_{\Pat}(s,a) = \set{s}$
		\end{itemize}
		
		\item If $op(a) = \nop$ then $\Delta^p_{\Pat}(s,a) = \set{s}$ for all $s$.
		
		\item If $a = \dummy$ then $\Delta^p_{\Pat}((k,neutral),a) = \Delta^p_{\Pat}((\infty,neutral),a) = \set{(\infty, neutral)}$
		
		\item Otherwise $\Delta^p_{\Pat}(s,a) = \es$.
		
		\item We add an $\epsilon$-transition from $(k, neutral)$ to $(\infty, neutral)$. It can be eliminated by adding a few transitions to the automaton, but we allow it as it simplifies the proof.
		
	\end{itemize}
	
	The acceptance condition $\phi_{\Pat}$ is simply $\varinf_{(\infty,neutral)}$.
	
	Let $\run$ be a local run of $p$ matching the given "stair pattern" $\Pat$, and let  $\run=\run_{0} a_{1} \cdots \run_{k-1} a_{k} \run_{k} \run_{k+1} \cdots$ be its "stair decomposition".
	For all $i < k$ there is a path in the automaton reading $\run_{p,i}$ from state $(i, neutral)$ to itself: every letter acquiring some $t$ (thus getting to state $(i,t)$) is later followed by one releasing it. Letters using a lock lower than $t_i$ for $\leq_p$ cannot appear in $\run_{p,i}$ as otherwise $\run_p$ would not match $\Pat$. 
	Furthermore, there are no $\dummy$ in $\run_{p,i}$. As a result, when $t$ is eventually released we are still in state $(i, t)$ and we go back to state $(i, neutral)$.
	
	As a result, the run $\run_{0} a_{1} \cdots \run_{k-1} a_{k} \run_k$ labels a path from $(0, neutral)$ to $(k,neutral)$ in the automaton. 
	Then all letters that appear in the $\run_i$ for $i \geq k$ are greater than $t_k$, otherwise $\run_p$ would not match $\Pat$. If $\run_p$ is finite then all the following letters are $\dummy$, and we stay in $(\infty, neutral)$ forever.
	
	If $\run_p$ is infinite then by definition of the "stair decomposition" all the following letters use locks of $\InfN_p$ or apply $\nop$. Then we can take the $\epsilon$ transition to $(\infty,neutral)$. Each $w_{j}$ with $j \geq i+1$ labels a path from $(\infty, neutral)$ to itself: all $\get{t}$ operations that get the run to $(\infty,t)$ are matched by a later operation $\rel{t}$ taking it back to $(\infty, neutral)$. 
	In both cases the run is accepting as it visits $(\infty,neutral)$ infinitely many times.
	
	Now let $\run$ be a "nested" local run of $p$ such that $\pad{\run}$ is accepted by $\aut_{\Pat}$. Then we consider an accepting computation of $\run$ in $\aut$ and decompose $\run$ as $\run=\run_{0} a_{1} \cdots \run_{k-1} a_{k} \run_{\infty}$ with $a_i$ the first letter in the run such that the computation gets to $(i,neutral)$ after reading the prefix $\run_{0} a_{1} \cdots \run_{i-1} a_{i}$.
	By definition of the automaton, we must have $op(a_i) = \get{t_i}$ for all $i$, and all operations executed after $a_i$ must be on locks greater than $t_i$ for $\leq_p$.

	We show that for all $i \in \set{0,\ldots,k, \infty}$, any "nested" run labelling a path from $(i, neutral)$ to itself must be "neutral":
	Suppose it is not the case, let $u$ be a "nested" run that is not "neutral" labelling a path from $(i,neutral)$ to itself, of minimal size.
	The first operation on locks in $u$ must be a $\get{t}$, as otherwise $u$ cannot be read from $(i, neutral)$. In order to go back to $(i, neutral)$, there must be a later operation $\rel{t}$ in $u$. Hence, as $u$ is "nested", we have $u = a v b u'$ with $op(a) = \get{t}$, $op(b) = \rel{t}$, $v$ "neutral", and $u'$ labelling a path from $(i,neutral)$ to itself. By minimality of $u$, $u'$ must be "neutral", hence so must be $u$: contradiction.
	
	All runs $\run_i$ label a path from $(i, neutral)$ to itself, thus they must be "neutral".
	
	If $\run$ is finite, then so is $\run_{\infty}$. Furthermore $\run_{\infty}$ is "neutral" as it must label a path from $(k,neutral)$ to itself. Thus $\run$ must match the "stair pattern" $\Pat$.
	
	Otherwise, we cut $\run_{\infty}$ in parts so that  $\run_{\infty}= \run_k \run_{k+1} \cdots$ with $\run_k$ labelling a path from $(k, neutral)$ to itself and for all $j \geq k+1$, $\run_j$ labelling a path from $(\infty,neutral)$ to itself. 
	
	This decomposition exists as, for $\run$ to be accepted, $\run_{\infty}$ must label a path starting in $(k, neutral)$, taking at some point the $\epsilon$ transition from $(k, neutral)$ to $(\infty, neutral)$, and then going back infinitely many times to $(\infty, neutral)$.
	
	As each $\run_j$ labels a path from either $(k, neutral)$ or $(\infty,neutral)$ to itself, they are all "neutral". Furthermore, the $\run_j$ with $j \geq k+1$ can only use locks of $\InfN_p$, as the automaton only allows those operations from states with an $\infty$ first component.
	
	As a result, $\run$ matches pattern $\Pat$. 
\end{proof}

We can finally give an \NP~upper bound for the problem over "nested" "LSS".

\begin{proposition}
	The "regular verification problem" is decidable in \NP~for "sound" "nested" "LSS".
\end{proposition}

\begin{proof}
	Let $\tss = ((\aut_p)_{p \in \Proc},T)$ be a "sound" "nested" "LSS", and $((\objaut{p})_{p \in \Proc}, \phi)$ a "regular objective".
	
	The algorithm is similar to the one for Proposition~\ref{prop:NP-2LSS}: 
	we guess a pattern $\Pat_p =(\OwnsN_{p}, \leq_p, \InfN_p)$ for each process $p$, and a valuation $\nu$ of the variables $(\varinf_{p,s})_{p\in \Proc, s \in S_{\objaut{p}}}$ (the variables of $\phi$, see Definition~\ref{def:reg-obj}). We check that $\nu$ satisfies $\phi$.
	We transform each $\aut_p$ into $\aut_p^\dummy$, in which we added, on each state whose outgoing transitions all acquire a lock from $\bigcup_{p \in \Proc} \OwnsN_{p}$, a $\dummy$ self-loop.
	We also equip each $\objaut{p}$ with an Emerson-Lei accepting condition expressing that the run matches $\nu$.
	
	We then guess, for each process $p$, a run in the product of $\objaut{p}$, $\aut_p$ and $\aut_{\Pat_p}$ (as described in Lemma~\ref{lem:DELA-pat-nested}) that matches valuation $\nu$.
	It is folklore that if an Emerson-Lei automaton has an accepting run then it has one of the form $uv^\omega$ with $u$ and $v$ of polynomial size in the number of states of the automaton.
	Thus we can guess an accepting run within \NP.
	An accepting run is one that respects $\nu$ in $\objaut{p}$, and follows a run of $\aut_p$ of pattern $\Pat_p$ (and ends in a state with all outgoing transitions getting a lock of $\bigcup_{p \in \Proc} \OwnsN_{p}$ if it is finite).
	
	By Lemma~\ref{lem:pattern-char-nested}, we accept if and only if there is a "process-fair" global run of the "LSS" satisfying the objective. 
\end{proof}

We give a matching lower bound, robust to many restrictions. The reduction also solves a question left open in~\cite{Brotherson21}, as explained at the beginning of the section.

\begin{proposition}
	The "process deadlock problem" and the "circular deadlock" problem are \NP-hard for "sound" "nested" "exclusive" "LSS".
\end{proposition}

\begin{proof}
	We reduce the Independent Set Problem, in which we are given an undirected graph $G = (V,E)$ (edges are subsets of $V$ of size $2$) and an integer $k$ and have to determine whether there is a subset of vertices $S \subseteq V$ such that $\size{S} = k$ and there are no edges between any two elements of $S$. Let $n = \size{V}$.
	
	Let $G = (V, E)$ be an undirected graph, and $k \in \nats$. 
	We can assume that $V = \set{1, \ldots, n}$ for some $n \in \nats$.
	We set $E = \set{e_1, \ldots, e_m}$, i.e., we put an arbitrary order on edges in $E$.
	Our set of processes is $\Proc = \set{p_1, \ldots, p_k}$. 
	For each $1 \leq j \leq m$ we have a lock $t_j$. We write $T$ for the set $\set{t_j \mid 1 \leq j \leq m}$.
	Our set of locks is $T \cup \set{\ell_1, \ldots, \ell_k}$.
	For each $v \in V$ we write $E_v$ for the set of edges adjacent to $v$ and $T_v$ for $\set{t_j \mid e_j \in E_v}$. 
	Each process $p_i$ uses locks of $T \cup \set{\ell_i, \ell_{i+1}}$, with the convention $\ell_{k+1} = \ell_1$.
	
	Each process $p_i$ has $n$ transitions from its initial state, with operation $\nop$, which lead to states $s_1, \ldots, s_n$.
	From each $s_v$ a sequence of transitions (with no choice) acquires all locks $t_j \in T_v$ in increasing order of indices, then acquires $\ell_i$, then $\ell_{i+1}$, and then releases all those locks in reverse order (thus ensuring the nested property).
	We end up in a state $end_i$ with a local self-loop.
	This system is clearly "exclusive", as the only state with several outgoing transitions is the initial one, and none of them acquire any lock. 
	
	Suppose that this "LSS" has a run $\run$ leading to a circular deadlock. The structure of the "LSS" imposes that when executing $\run$ we eventually stay in the same configuration forever, with some processes blocked because they cannot acquire some lock and some looping indefinitely on their state $end_i$.
	
	Let $C$ be that configuration. If some $p_i$ is stuck after acquiring $\ell_i$, then it cannot have acquired $\ell_{i+1}$, as otherwise it could release all of its locks and loop in $end_i$. Hence some other process holds $\ell_{i+1}$, and it can only be $p_{i+1}$ (with $p_{k+1} = p_1$).
	By iterating this reasoning, we conclude that all processes $p_i$ are blocked while holding $\ell_i$, as they cannot acquire $\ell_{i+1}$.
	They must be holding disjoint sets of locks.
	By construction, each $p_i$ is holding $\ell_i$, plus the locks of some $T_v$, $v \in V$. Hence we have $k$ disjoint $T_v$, i.e., we have a set of $k$ vertices whose sets of adjacent edges are disjoint, i.e., an independent set of size $k$.
	
	Now suppose no $p_i$ is stuck after acquiring $\ell_i$. Then all $p_i$ that  have acquired $\ell_i$ have reached $end_i$, and released all their locks, thus all $\ell_i$ are free.
	There must be at least one process blocked when trying to acquire an element of some $T_v$. Let $j$ be the highest index in $\set{1, \ldots, m}$ such that there is a process $p_i$ blocked because it cannot acquire $t_j$. 
	Then there is a process $p_{i'}$ which is holding $t_j$, and is itself unable to acquire some $t_{j'}$ (as all locks $\ell_{r}$ are free). However, as all processes acquire elements of $T$ in increasing order of index, we must have $j' > j$, contradicting the maximality of $j$.
	Thus this case cannot happen, concluding the first part of our reduction.
	
	Conversely suppose we have an independent set of vertices $S = \set{v_1, \ldots, v_k} \subseteq V$ of size $k$. 
	Then we construct the run $\run$ in which, one by one, each $p_i$ first goes to $s_{v_i}$ and then acquires $\set{\ell_i}\cup T_{v_i}$.
	This is possible as they all acquire disjoint sets of locks.
	We end up in a configuration where each $p_r$ needs $k_{r+1}$ to advance, but cannot do so as $k_{r+1}$ is held by $p_{r+1}$.
	Hence $\run$ yields a "circular deadlock" (and even a "global deadlock", which shows that it is "process-fair"). 
	This ends our reduction, proving that the "circular deadlock" problem is \NP-hard even for "nested" "exclusive" "LSS".
	In the "LSS" above, we showed that if a run yields a "circular deadlock" then it yields a "global deadlock". Hence we can apply the reduction to the "process deadlock problem" by picking an arbitrary process $p_i$. There is a "process-fair" run with a finite projection on $p_i$ if and only if there is a solution to the initial Independent set problem. 
\end{proof}

\begin{remark}
	The Independent set problem is \NP-hard even on graphs of degree 3~\cite{GareyJS1976} (Theorem~2.6). As in the reduction above the number of locks used by each process is bounded by the degree of the input graph, we conclude that the lower bound still holds for systems where each process uses at most 5 locks.
\end{remark}

\section{Conclusion}

We have studied  the verification problem for "LSS" against boolean combinations of regular local objectives.
We established \PSPACE-completeness for the general problem, and presented two subcases where the verification problem becomes
\NP-complete: "2LSS" and "nested" "LSS", as well as a \PTIME~algorithm for the "process deadlock
problem" for "exclusive" "2LSS".
The \NP~and \PTIME~upper bounds use  as their main ingredient the characterisations of whether local runs can be scheduled into global ones through "patterns".
All lower bounds are robust, as they hold with bounds on the number of locks per process and very simple objectives.

Concerning future work, most of our results can easily be extended to the case
when processes are pushdown systems (except for the general case, which is
undecidable instead of \PSPACE-complete, see~\cite{KahIvaGup05}, Theorem~8). 
Another easy extension is to replace "nested" with bounded lock chains, a weaker condition defined in \cite{Kahlon09}.
These essentially do not require new ideas, thus we chose to not include them to avoid unnecessary details and highlight the key ingredients.
At the time of writing this paper, we are working towards implementing the algorithms described here (using a SAT solver for the \NP-hard problems), in which we plan to include those extensions.

About open problems, we do not know if "partial deadlocks" can be detected in \PTIME~for "exclusive" "LSS", or for "2LSS"  (not necessarily "exclusive"). Probabilistic algorithms have proven useful in distributed systems (see, for instance, the Lehmann-Rabin algorithm~\cite{LehmannR1981}), hence one may want to add probabilities to the model. Finally, versions of the problem with parameterized number of processes or locks could be of interest. 

\emph{I would like to thank Anca Muscholl and Igor Walukiewicz for their support and useful comments.}

\bibliographystyle{splncs04}
\bibliography{m.bib}

\newpage
\appendix
\appendixtrue

\end{document}